\documentclass[twocolumn,english,pre,floatfix,citeautoscript,nofootinbib]{revtex4}
\usepackage{amsbsy}
\usepackage{latexsym,epsfig,graphicx}
\usepackage{dcolumn}
\usepackage{subfigure}
\usepackage{comment}
\usepackage{color}
\usepackage[colorlinks,urlcolor=blue,citecolor=blue]{hyperref}
\usepackage{amstext}
\usepackage{amssymb}
\usepackage{setspace}
\usepackage{amsmath}
\usepackage{makeidx}
\usepackage{amsthm}
\usepackage{bm}
\usepackage{ulem}
\usepackage{multirow}
\usepackage{mathrsfs}

\usepackage{array}

\newcommand{\ii}{\mathrm{i}}
\newcommand{\e}{\mathrm{e}}
\newcommand{\dd}{\mathrm{d}}
\newcommand{\T}{\mathrm{T}}
\newcommand{\im}{\mathrm{Im}}
\newcommand{\re}{\mathrm{Re}}
\newcommand{\diag}{\mathrm{diag}}
\newtheorem{theorem}{Theorem}

\begin{document}

\title{Dirac monopole potentials with high charges underlying nonlinear waves }
\author{Yan-Hong Qin$^{1}$}
\author{Jin-Peng Yang$^{2}$}
\author{Li-Chen Zhao$^{3,4,5,6}$}\email{zhaolichen3@nwu.edu.cn}
\address{$^{1}$School of Physical Science and Technology, Xinjiang University, Urumqi, 830046, China}
\address{$^{2}$College of Mathematics and System Sciences, Xinjiang University, Urumqi 830046, China}
\address{$^{3}$School of Physics, Northwest University, Xi'an 710127, China}
\address{$^{4}$Shaanxi Key Laboratory for Theoretical Physics Frontiers, Xi'an 710127, China}
\address{$^{5}$Peng Huanwu Center for Fundamental Theory, Xi'an 710127, China}
\address{$^{6}$Fundamental Discipline Research Center for  Quantum Science and technology of Shaanxi Province, Xi'an 710127, China}
\begin{abstract}
We investigate topological vector potentials underlying the phases of nonlinear waves by performing Dirac's magnetic monopole theory in an extended complex plane, taking into account self-steepening effects while ignoring the usual cubic nonlinearities. We uncover that the simple poles and third-order poles of the density function constitute virtual monopole fields with higher charges $\pm3/2$ and $\pm5/2$, respectively. These results are in sharp contrast to the previous findings, where the simple zeros of the density function yield charges $\pm1/2$.  We choose scalar and vector rogue waves as well as bright solitons to demonstrate the Dirac monopole potentials. These results confirm a series of quantized magnetic charges for virtual monopoles underlying nonlinear waves, and reveal new relations between poles of density functions and topological charges.

\end{abstract}
\date{\today}
	
\maketitle
	
\section{Introduction}

Magnetic monopoles, first proposed by Dirac in the context of quantum mechanics \cite{Dirac1}, have long intrigued physicists due to their profound implications for quantum mechanics \cite{Quantum1,Quantum2} and charge quantization \cite{Dirac1,monopole1}. Dirac monopoles are predicted by analyzing characteristics of zeros of wave functions and the phase gradient properties in three-dimensional real space \cite{Dirac1}. Dirac demonstrated that the monopoles are located at the endpoints of nodal lines and admit quantized charges $n/2$ ($n$ is an integer) \cite{Dirac1}. It is of great significance to extend the Dirac monopole theory to other spaces, such as momentum \cite{momentum1,momentum2,momentum3,momentum4}, parameters \cite{parameter1,parameter2,parameter3,parameter4} and even complex space \cite{Zhao5,Zhao3}. Particularly, density zeros in the complex plane correspond to singularities of the topological vector potential, which can be identified as Dirac's virtual magnetic monopole fields, and the charges of these virtual monopoles are uncovered to admit general characteristics of Dirac monopoles. Therefore, in the framework of Dirac monopole theory, virtual monopoles in the complex plane have been used to address phase singularity problems of one-dimensional wavefunctions, including nonlinear waves such as rogue waves (RWs) \cite{Zhao5,Zhao2}, dark solitons \cite{Zhao3,Qin3}, bright solitons (BSs) \cite{Zhao1}, and other localized waves \cite{Zhao2,Zhao4}, as well as linear quantum-well eigenstates \cite{Wang2024} and linear and nonlinear wave-field interference \cite{Sun2023}.

For these localized waves, the associated virtual monopoles all carry charges of $\pm 1/2$ in the complex plane \cite{Zhao5,Zhao2,Zhao1,Zhao4,Zhao3,Qin3}, including dark solitons \cite{Zhao3}, first- to third-order RWs \cite{Zhao5,Zhao2}, and Akhmediev breathers \cite{Zhao4} of the nonlinear Schr\"{o}dinger equation; RWs and rational W-shaped solitons of the Hirota and Sasa-Satsuma equations \cite{Zhao2}; dark solitons of the Hirota equation \cite{Qin3}; and BSs of the generalized Chen-Lee-Liu equation \cite{Zhao1}. These studies cover conventional cubic Kerr nonlinearity, third-order dispersion, delayed nonlinear response, and derivative self-steepening. Distinct charges $\pm1$ were found by tracking density zeros using Dirac's theory across different parameter spaces, rather than in the complex plane \cite{Zhao4}. However, high odd multiple magnetic charges are still not found in the complex plane, partly due to a lack of a profound understanding of the quantitative relations between the density function and topological singularities. As far as we know, the previously reported charges $\pm 1/2$ all correspond to the simple zeros of the density functions. It could be expected that multiple zeros of density can induce high-charge monopoles from Dirac monopole field studies \cite{Ray2014,Pietila,Ray2015,JHZheng}, but it is hard to find multiple zeros of the density in the complex plane.

In this work, we show that high Dirac monopole charges $\pm3/2$ and even $\pm5/2$ exist for nonlinear waves with self-steepening effects but without the usual cubic nonlinearities. We extend Dirac's magnetic monopole theory to the complex plane, applying this framework to RW and BS solutions in both scalar and vector derivative nonlinear Schr\"{o}dinger (DNLS) systems. For RWs, we find pair charges of $\pm1/2$ and $\pm3/2$ in both the scalar and vector DNLS systems.  Specifically, the density poles of a scalar single-hump BS (SHBS) can induce singularities of topological vector potentials. These singularities behave as periodically distributed point-like virtual monopoles carrying the high topological charge of $\pm3/2$ and a quantized flux of $\pm3\pi$. More remarkably, we uncover that vector double-hump BSs (DHBSs) admit various topological charges of $\pm1/2$, $\pm3/2$, and $\pm5/2$, which correspond precisely to simple zeros, simple poles, and third-order poles of density functions, respectively. These results are in sharp contrast to those for nonlinear waves with cubic nonlinearities, and even for BSs with the cubic term $|q|^2q$ and the self-steepening effect $|q|^2q_x$ \cite{Zhao1}, where only virtual monopoles of charge $\pm 1/2$ can be found. The present DNLS equations instead contain a total-derivative self-steepening effect of the form $(|q|^2q)_x$, or its vector counterpart, without an independent cubic term $|q|^2q$. We for the first time show that the poles of the density function in the extended complex plane can induce the singularities of vector potentials, in contrast to Dirac's usual monopole theory. Further analysis suggests that high odd multiple magnetic charges are induced by high-order poles of the density function in the complex plane, and high even multiple magnetic charges could be found via high-order even poles or multiple zeros of density functions.

The remainder of this paper is organized as follows. In Sec.~\ref{sec2}, we analyze the phase of the scalar eye-shaped (ES) RW via the topological vector potential. Its solution contains four pairs of virtual magnetic monopoles: two with charge $\pm3/2$ and two with charge $\pm1/2$. The latter collide and merge when valleys emerge. Two mergers result in a time interval during which the phase shift $\Delta\phi=2\pi$. In Sec.~\ref{sec3}, we extend our study to the vector system, focusing on the four-petaled (FP)-ES RW. For vector RWs, we uncover that the monopoles carrying $\pm1/2$ charges merge on the imaginary axis in FPRWs and on the real axis in ESRWs, both yielding zero net phase shift. In Sec.~\ref{sec4}, we demonstrate that the SHBS in the scalar DNLS system hosts periodically distributed monopoles with a high topological charge of $\pm3/2$ and a quantized flux of $\pm3\pi$. This structure allows the soliton's phase shift to vary over an unusual range of $(0, 3\pi)$. Even more remarkably, in Sec.~\ref{sec5}, we analytically investigate the topological phase of DHBSs in the vector DNLS system. We uncover a family of virtual magnetic monopoles with quantized charges of $\pm1/2$, $\pm3/2$, and $\pm5/2$. We establish that simple zeros, simple poles, and third-order poles of the density function correspond precisely to these distinct charge values. Typical examples are illustrated by considering three types of static symmetric DHBSs, as well as the collision of two DHBSs. Results are summarized and discussed in Sec.~\ref{sec6}.

\section{topological phase of scalar rogue wave}\label{sec2}

We begin with a scalar system to study the Dirac monopole vector potentials of nonlinear waves, taking into account self-steepening effects while ignoring the usual cubic nonlinearities; this is the DNLS given by \cite{ref3}
\begin{align}\label{model}
\ii q_t+q_{xx}+\frac{2}{3}\ii(|q|^2q)_x=0,
\end{align}
where the last term accounts for self-steepening effects (also known as shocking). This equation governs the evolution of slowly varying wave amplitudes and has been widely studied across plasma physics \cite{ref1,ref2,ref3,Mio,Plasma1,Plasma2,Plasma3,Plasma4,Plasma5}, optical fibers \cite{optic1,optic2,optic3,optic4}, and ferromagnetic media \cite{FM1,FM2,FM3}. The physical interpretations of $q(x,t)$, $x$ and $t$ vary across these contexts. Specifically, in magnetized plasmas, $q(x,t)$ is the normalized complex envelope of an Alfv\'{e}n wave, $x$ denotes the spatial coordinate along the background magnetic field, and $t$ corresponds to the slow temporal scale. In optical fibers for femtosecond pulses, $q(x,t)$ represents a normalized complex amplitude of the pulse envelope, $t$ is a normalized distance along the fiber, and $x$ is the normalized time within the frame of reference moving along the fiber at the group velocity. In a ferromagnetic medium, $q(x,t)$ represents the slowly varying envelope of the electromagnetic wave propagation in the ferromagnetic media, $t$ represents the slow temporal evolution of the pulse envelope and $x$ describes the slow spatial distribution of the electromagnetic wave pulse envelope along the propagation direction. As an integrable system, the DNLS equation admits a variety of localized wave solutions, such as BSs, RWs, and semirational solutions \cite{Kaup,Ling,He1,Xu,Zuo1,Chan,Min1}, which can be derived in this integrable system using the inverse scattering method \cite{Kaup}, the Darboux transformation \cite{Ling,He1,Xu,Zuo1}, and Hirota methods \cite{Chan,Min1}.

For nonlinear waves with the usual cubic nonlinearities, virtual monopole charges are always $\pm 1/2$ \cite{Zhao5,Zhao3,Zhao4}. The same holds when self-steepening effects are added to cubic nonlinearities \cite{Zhao1}. The DNLS system considered here, by contrast, retains self-steepening while completely discarding the usual cubic nonlinearities. We here investigate whether the exact solutions of this system can support higher topological charges. Among the various localized solutions, RWs are particularly suitable for a first analysis, as they are spatiotemporally doubly localized and admit only a finite number of monopole pairs \cite{Zhao5}, making the topological analysis more explicit and tractable than for other localized waves with periodic monopole distributions. We therefore first investigate the topological phase of RWs within the scalar DNLS system. The RW solutions of system \eqref{model} can be obtained by performing a Darboux transformation on the seed solution $q_0\!=\!s\e^{-\ii k[x-(s^2-k)t]}$ and are expressed as \cite{Lin}
\begin{align}\label{rwsol}
q=\sqrt{\frac{3}{2}}\left[\ii\frac{q_0}{ k}\frac{\chi_0(\chi_0^*+k)}{\chi_0^*(\chi_0+k)}\frac{1+s^2r^2\widehat{\rho}}{1+s^2r^2\rho}\right]_x,
\end{align}
where $r=\im(\sqrt{s^2-2k})$, $\chi_0=-\frac{r^2+\ii sr}{2}$, $\rho=u\left(u^*+\frac{1}{\chi_0^*}\right)$, $\widehat{\rho}=\left(u+\frac{1}{\chi_0}-\frac{1}{\chi_0+k}\right)\left(u^*+\frac{1}{\chi_0^*+k}\right),$
and $u=\ii \frac{s}{r}\left[-x+(s^2-r^2)t-\frac{1}{sr}\right]+srt.$
To facilitate precise analysis, we perform the following coordinate transformation $(x,t)\mapsto (\pmb{x},\pmb{t})$,
$$\pmb{x}=-x+(s^2-r^2)t-\frac{1}{sr}-\frac{1}{2k},~\pmb{t}=sr t+\frac{1}{2k}.$$
The intensity of solution \eqref{rwsol} exhibits three extreme points on the $(\pmb{x},\pmb{t})$ plane: a peak at $p_0=(0,0)$ and two valleys at $p_1=\left(-\frac{\sqrt{3}(2+\mu^2)}{\tau},\pmb t_1\right), p_2=\left(\frac{\sqrt{3}(2+\mu^2)}{\tau},\pmb t_2\right)$,
with $\pmb t_1\!=\!-\pmb t_2\!=\!-\frac{\sqrt{3}\mu^2}{\tau}$, $\mu\!=\!s/r$ and $\tau\!=\!s^2(1\!+\!\mu^2)\sqrt{4\!+\!\mu^2}>0$. Thus, the solution \eqref{rwsol} exclusively exhibits the ESRW pattern, as shown in Fig. \hyperref[fig1]{1(A)}. Since the transformation on $t$ involves only constant translation and scaling, the peak and valleys cannot appear together in any RW of this type, in contrast to RWs in standard NLS systems \cite{Akhmediev}.

To study the topological phase of the RW, we first give its phase gradient $\phi_{\pmb{x}}(\pmb{x},\pmb{t})=\frac{F(\pmb{x},\pmb{t})}{\prod_{i=1}^{4}[\pmb{x}-\pmb{x}_+^{[i]}(\pmb{t})][\pmb{x}-\pmb{x}_-^{[i]}(\pmb{t})]}$ \cite{Dirac1},
where the expression of $F(\pmb{x},\pmb{t})$ is given in Appendix \ref{app2}. Here, we consider $\mu>0$ without loss of generality. The phase distribution of the RW in the $(\pmb{x},\pmb{t})$ plane is obtained via the variable-upper-limit integral $\phi (\pmb{x},\pmb{t})=\int_{-\infty}^{\pmb{x}}\phi_{\pmb{x}} \mathrm{d} \pmb{x}$, as shown in Fig. \hyperref[fig1]{1(B)}, where abrupt phase jumps are clearly visible in the regions of the central hump and two valleys. Moreover, the singularities of the phase gradient correspond strikingly to both the zeros and poles of the density in this RW solution. Indeed, the density is
$|q|^2=\frac{3}{2}s^2\frac{\prod_{i=1}^{2}[\pmb{x}-\pmb{x}_+^{[i]}(\pmb{t})][\pmb{x}-\pmb{x}_-^{[i]}(\pmb{t})]}
{\prod_{i=3}^{4}[\pmb{x}-\pmb{x}_+^{[i]}(\pmb{t})][\pmb{x}-\pmb{x}_-^{[i]}(\pmb{t})]}$, while the phase gradient $\phi_{\pmb{x}}$ has the denominator $\prod_{i=1}^{4}[\pmb{x}-\pmb{x}_+^{[i]}(\pmb{t})][\pmb{x}-\pmb{x}_-^{[i]}(\pmb{t})]$. Comparing the two, the singularities for $i=1,2$ correspond to simple zeros of the density function, whereas those for $i=3,4$ correspond to simple poles of the density function. Previous studies have linked simple zeros of density to monopole fields with quantized flux $\pi$ \cite{Zhao3,Zhao5}. The emergence of density poles in the present solution naturally necessitates an investigation of their topological properties. We thus examine the topological vector potential associated with these singularities in the complex plane.

Detailed analysis of $\phi_{\pmb{x}}$ identifies eight complex-valued singularities:
$\pmb{x}_\pm^{[1]}=\pm\ii \frac{\mu}{\xi}-\frac{\sqrt{\zeta_1\mp\ii \delta_1 \pmb{t}-\xi^2 \pmb{t}^2}}{\mu\xi},
\pmb{x}_\pm^{[2]}=\pm\ii \frac{\mu}{\xi}+\frac{\sqrt{\zeta_1\mp\ii \delta_1 \pmb{t}-\xi^2 \pmb{t}^2}}{\mu\xi},
\pmb{x}_\pm^{[3]}=\pm\ii \frac{\mu}{\xi}-\frac{\sqrt{\zeta_2\pm\ii \delta_2 \pmb{t}-\xi^2 \pmb{t}^2}}{\mu\xi},
\pmb{x}_\pm^{[4]}=\pm\ii \frac{\mu}{\xi}+\frac{\sqrt{\zeta_2\pm\ii \delta_2 \pmb{t}-\xi^2 \pmb{t}^2}}{\mu\xi},$
with $\zeta_1=\mu^2(3+2\mu^2), \zeta_2=-\mu^2(1+2\mu^2), \xi=s^2(1+\mu^2), \delta_1=2\mu\xi(2+\mu^2), \delta_2=2\mu\xi^3$. Their complex nature necessitates analytic continuation from the real variable to the complex plane:  $\pmb{x} \mapsto z = \pmb{x} + \ii \pmb{y}$, and we accordingly extend the phase gradient $\phi_{\pmb{x}}(\pmb{x}, \pmb{t}) \mapsto \phi_z(z, \pmb{t})$.
These eight singularities then take the form of complex coordinates $z_n = \pmb{x}_n + \ii\pmb{y}_n$ ($n = 1, \cdots, 8$), where for each $i=1,\cdots,4$, $\pmb{x}_-^{[i]}$ and $\pmb{x}_+^{[i]}$ map to $n=2i-1$ and $n=2i$, with $\pmb{x}_n = \mathrm{Re}[\pmb{x}_{\pm}^{[i]}], \pmb{y}_n = \mathrm{Im}[\pmb{x}_{\pm}^{[i]}]$.

A complex vector potential is then introduced on the $(\pmb{x},\pmb{y})$-plane as
$\mathscr{A}_c=\frac{\partial\phi(z)}{\partial\pmb{x}}\textbf{e}_{\pmb{x}}+\frac{\partial\phi(z)}{\partial\pmb{y}}\textbf{e}_{\pmb{y}}$. The corresponding topological vector potential $\mathbf{A}=\mathrm{Re}[\mathscr{A}_c]$ is given by
$\mathbf{A}=\sum_{n=1}^{8}\frac{\Omega_n [(\pmb{x}-\pmb{x}_n)\mathbf{e}_{\pmb{y}}-(\pmb{y}-\pmb{y}_n)\mathbf{e}_{\pmb{x}}]}{2 \pi[(\pmb{x}-\pmb{x}_n)^2+(\pmb{y}-\pmb{y}_n)^2]}$ \cite{Zhao5,Zhao2,Zhao3}. Each $z_n$ in the complex plane acts as a virtual magnetic monopole within the vector potential fields. Their trajectories evolving with $\pmb{t}$ characterize dynamics of four virtual monopole pairs dominating topological properties of RWs. Each monopole carries a quantized flux $\Omega_n = 2\pi \ii \lim\limits_{z \to z_n^{\pm}} (z-z_n^{\pm}) \phi_z$. Specifically, $\Omega_n = \pm\pi$ for $n=1,\cdots,4$, corresponding to simple zeros of the density function with charges $\pm1/2$, while $\Omega_n = \pm3\pi$ for $n=5,\cdots,8$, for simple poles of the density function with charges $\pm3/2$. Thus, the topological vector potential for the scalar RW can be explicitly expressed as:
\begin{align}\label{TVP_sc_rw}
\mathbf{A}=&-\frac{1}{2}\sum\limits_{n=1,4}\frac{ (\pmb{x}-\pmb{x}_n)\mathbf{e}_{\pmb{y}}-(\pmb{y}-\pmb{y}_n)\mathbf{e}_{\pmb{x}}}{(\pmb{x}-\pmb{x}_n)^2+(\pmb{y}-\pmb{y}_n)^2}\nonumber\\
&+\frac{1}{2}\sum\limits_{n=2,3}\frac{ (\pmb{x}-\pmb{x}_n)\mathbf{e}_{\pmb{y}}-(\pmb{y}-\pmb{y}_n)\mathbf{e}_{\pmb{x}}}{(\pmb{x}-\pmb{x}_n)^2+(\pmb{y}-\pmb{y}_n)^2}\nonumber\\	
&-\frac{3}{2}\sum\limits_{n=5,8}\frac{
(\pmb{x}-\pmb{x}_n)\mathbf{e}_{\pmb{y}}-(\pmb{y}-\pmb{y}_n)\mathbf{e}_{\pmb{x}}}{(\pmb{x}-\pmb{x}_n)^2+(\pmb{y}-\pmb{y}_n)^2}\nonumber\\
&+\frac{3}{2}\sum\limits_{n=6,7}\frac{
(\pmb{x}-\pmb{x}_n)\mathbf{e}_{\pmb{y}}-(\pmb{y}-\pmb{y}_n)\mathbf{e}_{\pmb{x}}}{(\pmb{x}-\pmb{x}_n)^2+(\pmb{y}-\pmb{y}_n)^2}.
\end{align}
The corresponding magnetic field will be zero everywhere except at those singular points; that is, $\mathbf{B} = \nabla \times \mathbf{A} = \mathbf{e}_z\sum\limits_n \Omega_n \delta (\pmb{x} - \pmb{x}_n, \pmb{y} - \pmb{y}_n)$. A proof of Eq.~\eqref{TVP_sc_rw} is given in Appendix \ref{app0}. It should be noted that as $\pmb{t}$ crosses $0$, $\pmb{x}_\pm^{[3]}$ and $\pmb{x}_\pm^{[4]}$ cross the branch cut. A fixed branch thus leads to discontinuity; hence, switching branches at $\pmb{t}=0$ is necessary for continuity. Consequently, in the scalar RW under self-steepening, the topological vector potential $\mathbf{A}$ in Eq.\eqref{TVP_sc_rw} contains fields constituted by relevant simple zeros and simple poles of the density function, carrying charges $\pm1/2$ and $\pm3/2$, respectively. This yields a topological structure richer than previously reported cases, where only $\pm1/2$ charges appear \cite{Zhao5,Zhao2,Zhao1,Zhao4,Zhao3,Qin3}.

\begin{figure*}[htpb]	
\centering
\includegraphics[width=170mm]{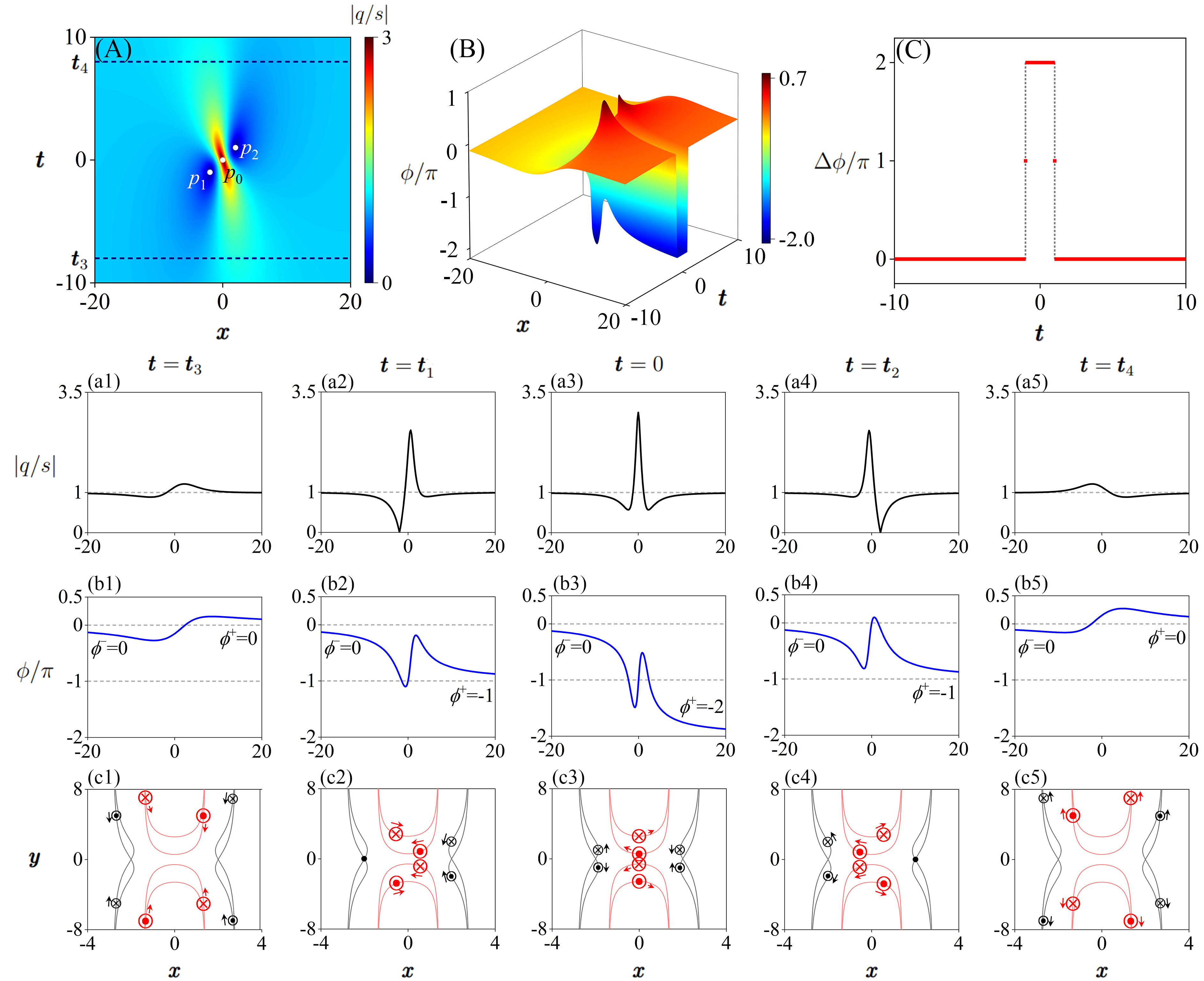}
\caption{Topological phase of ESRW. A, B, and C show the ESRW density, the spatiotemporal distribution of phase $\phi$, and the temporal variation of the phase shift $\Delta\phi$, respectively. Panels (a), (b), and (c) depict the density profile, the phase distribution, and the magnetic field of the RW at the corresponding times, respectively.  In (c), the trajectories of virtual monopoles with charges $\pm\frac{1}{2}$ (black curves) and $\pm\frac{3}{2}$ (red curves) are plotted, where the symbols $\bigodot$ and $\bigotimes$ denote positive and negative signs, respectively.
Parameters are: $s=\sqrt{\frac{\sqrt{2}}{3}},\mu=\sqrt{2}.$}\label{fig1}
\end{figure*}

Specifically, in the limit $\pmb{t} \to \pm \infty$, $\pmb{x}_\pm^{[1]}$ and $\pmb{x}_\pm^{[2]}$ describe four virtual monopoles that evolve from negative (or positive) infinity toward positive (or negative) infinity along the imaginary axis; in contrast, $\pmb{x}_\pm^{[3]}$ and $\pmb{x}_\pm^{[4]}$ represent another four monopoles that evolve from negative (or positive) infinity back to negative (or positive) infinity in the upper and lower half-planes. Particularly,  $\im\pmb{x}_\pm^{[1]}=0$ and $\im\pmb{x}_\pm^{[2]}=0$ occur at  $\pmb{t}_1$ and $\pmb{t}_2$, respectively, indicating that the two pairs of virtual monopoles described by $\pmb{x}_\pm^{[1]}$ and $\pmb{x}_\pm^{[2]}$ collide and merge on the real axis separately at $\pmb{t}_1$ (first valley emergence) and $\pmb{t}_2$ (second valley emergence), respectively. The monopole collision process can be well described by limit analysis on the vector potential.

As $\pmb{t}\rightarrow\pmb{t}_1$, we have $\pmb{x}_1=\pmb{x}_2$ and $\pmb{y}_1=-\pmb{y}_2\rightarrow 0$. By the dominated convergence theorem, for any test function $\varphi(\pmb{x})$,
$\lim\limits_{\pmb{y}_1\rightarrow \pm 0}\int_{-\infty}^{\infty}\frac{ \pmb{y}_1\varphi(\pmb{x})\dd\pmb{x}}{(\pmb{x}-\pmb{x}_1)^2+\pmb{y}_1^2}=\pm\pi\varphi(\pmb{x}_1)
	=\pm\pi\int_{-\infty}^{\infty}\delta(\pmb{x}-\pmb{x}_1)\varphi(\pmb{x})\dd\pmb{x}$ holds in the distributional sense, where $\delta(\pmb{x})$ denotes the Dirac delta function. Then, the vector potential $\mathbf{A}$ in Eq.~\eqref{TVP_sc_rw} can be rewritten as
\begin{align*}
\lim_{\pmb{t}\to \pmb{t}_{1}\pm 0}\mathbf{A}
=&\mp\pi\delta(\pmb{x}-\pmb{x}_1)\mathbf{e}_{\pmb{x}}\!+\!\!\sum_{k \in \{3,5,7\}}\! \frac{c_k\pmb{y}_{k}\mathbf{e}_{\pmb{x}}}{(\pmb{x}\!-\!\pmb{x}_{k})^2\!+\!\pmb{y}_{k}^2},
\end{align*}
with $c_3 = 1$, $c_5 = -3$, $c_7 = 3$.
The second part has no singularities in the limit, and its integral over $\mathbb{R}$ is $-\pi$. Besides, the total integral at $\pmb{t}=\pmb{t}_1$ is given solely by the cumulative sum term since the singularity $\pmb{x}_\pm^{[1]}$ is absent at this instant, yielding the value $-\pi$. Analogously, for $\pmb{t}\rightarrow \pmb{t}_2$, we have $\pmb{x}_3=\pmb{x}_4$ and $\pmb{y}_3=-\pmb{y}_4\rightarrow0$, leading to
\begin{align*}
\lim_{\pmb{t}\to \pmb{t}_2\pm 0}\mathbf{A}
=&\pm\pi\delta(\pmb{x}-\pmb{x}_3)\mathbf{e}_{\pmb{x}}\!+\!\!\sum_{k\in \{1,5,7\}}\! \frac{c_k\pmb{y}_{k}\mathbf{e}_{\pmb{x}}}{(\pmb{x}\!-\!\pmb{x}_{k})^2\!+\!\pmb{y}_{k}^2},
\end{align*}
with $c_1=-1$. The cumulative sum also contains no singularities in the limit, and its integral over $\mathbb{R}$  also equals $-\pi$. As a result, the phase shift $\Delta\phi$ can be obtained by
\[\Delta\phi=-\int_{-\infty}^{\infty}\mathbf{A}\cdot\mathbf{e}_{\pmb{x}}\dd\pmb{x}=
\left\{
	\begin{aligned}
		&0,&&\pmb{t}<\pmb{t}_1~~\mathrm{or}~~\pmb{t}>\pmb{t}_2,&\\
		&\pi,&&\pmb{t}=\pmb{t}_1~~\mathrm{or}~~\pmb{t}=\pmb{t}_2,&\\
		&2\pi,&&\pmb{t}\in(\pmb{t}_1,\pmb{t}_2).&
	\end{aligned}
	\right.
\]

Following the above analysis, we characterize the phase properties of the RW at several key moments to facilitate physical insight, namely, during its emergence ($\pmb{t}_1, 0, \pmb{t}_2$) and before and after its appearance ($\pmb{t}_3, \pmb{t}_4$), as depicted in Fig. \hyperref[fig1]{1}. The RW profiles are shown in Figs. \hyperref[fig1]{1(a1)-(a5)}, with the corresponding phase distributions presented in Figs. \hyperref[fig1]{1(b1)-(b5)}. The underlying topological magnetic fields encoded in these phase patterns are illustrated in Figs. \hyperref[fig1]{1(c1)-(c5)} based on the topological vector potential Eq.~\eqref{TVP_sc_rw}, where black and red curves trace the trajectories of virtual monopoles carrying charges $\pm1/2$ and $\pm3/2$, respectively, with $\odot$ ($\otimes$) denoting positive (negative) charges. At each instant, the position of each monopole is uniquely defined. Before the first valley emerges ($\pmb{t}<\pmb{t}_1$), a pair of $\pm1/2$ monopoles in the left half-plane moves toward the real axis along opposite trajectories. At $\pmb{t}=\pmb{t}_1$, they collide and merge on the real axis, inducing an abrupt phase shift $0\rightarrow\pi$. After this elastic collision, the pair reemerges and continues along their respective trajectories, further advancing the phase shift $\pi\rightarrow2\pi$ before the second valley appears. When the RW hump emerges ($\pmb{t}=0$), two $\pm1/2$ monopole pairs are symmetrically displaced about the $\pmb{y}$-axis, while two $\pm3/2$ monopole pairs lie on the $\pmb{y}$-axis. When the second valley emerges ($\pmb{t}=\pmb{t}_2$), another $\pm1/2$ pair in the right half-plane meets and merges on the real axis, producing an opposite effect that reduces the abrupt phase shift from $2\pi\rightarrow\pi$. Subsequently, these monopoles cross the real axis, reappear, and evolve further, eventually returning the total phase shift $\pi\rightarrow0$. The overall phase shift evolution is shown in Fig. \hyperref[fig1]{1(C)}.

Topological vector potential theory has also been used to uniquely determine phase evolution in other nonlinear waves when abrupt changes make direct phase representations ambiguous, as demonstrated in Ref.~\cite{Zhao2} by the identification of the phases of rational W-shaped solitons through the RW limit in the Hirota and Sasa-Satsuma models and by the self-consistent analysis of second-order RW phase evolution, which resolved the $0$-versus-$2\pi$ ambiguity in experimental phase data \cite{Gxu}.

\section{Topological phase of vector rogue waves}\label{sec3}

Having revealed the rich topological charges and multiple monopole collision dynamics of scalar ESRWs, we now turn to the vector regime, where the interplay between components yields richer RW patterns, such as the anti-eye-shaped RW \cite{defocusing1,lingrw1,zhaorw1} and the FPRW \cite{zhaorw2,zhaorw3}. The coupled DNLS system including high-order effects (e.g., self-steepening) can produce novel dynamics, such as fundamental RWs with ultrastrong amplitude enhancement \cite{chen1,chen2}. Then, we take the FP-ESRW as a prototypical example to investigate its topological properties.

The vector DNLS system can be described by \cite{CDNLS1}
\begin{subequations}\label{model2}
\begin{align}
\ii q_{1,t}+q_{1,xx}+\frac{2}{3}\ii[(|q_1|^2+|q_2|^2)q_1]_x=0,\\
\ii q_{2,t}+q_{2,xx}+\frac{2}{3}\ii[(|q_1|^2+|q_2|^2)q_2]_x=0.
\end{align}
\end{subequations}
Its vector RW solution can be obtained by conducting a Darboux transformation \cite{Ling2,Lin} with the seed solution $q_{0,i}=s_i\e^{\ii k_i[-x+(s_1^2+s_2^2-k_i)t]}~(i=1,2)$ and is expressed as:
\begin{align}\label{Vec_rwsol}
	q_i=\sqrt{\frac{3}{2}}\left[\ii\frac{ q_{0,i}}{ k_i}\frac{\chi_0(\chi_0^*+k_i)}{\chi_0^*(\chi_0+k_i)}\frac{1+4(\im\chi_0)^2\rho_i}{1+4(\im\chi_0)^2\rho_0}\right]_x,
\end{align}
where $u=\ii\left[-x+(s_1^2+s_2^2+2\re\chi_0)t+\frac{1}{\im\chi_0}+\im\frac{\chi_2}{\chi_1^{2}}\right]-2\im\chi_0t+\re\frac{\chi_2}{\chi_1^{2}}$, $\rho_i\!=\!\left(u^*\!+\!\frac{1}{\chi_0^*\!+\!k_i}\right)\left(u\!+\!\frac{1}{\chi_0}\!-\!\frac{1}{\chi_0\!+\!k_i}\right),$ and $\rho_0\!=\!u\left(u^*\!+\!\frac{1}{\chi_0^*}\right).$
$\chi_j$ is given by the solution $\pmb\chi=\sum\limits_{j=0}^{\infty}\chi_j\epsilon^j$, and $\pmb\chi$ satisfies the equation
$\mathcal{P}(\pmb\chi;\eta)\equiv(\pmb\chi-2\eta)+\eta \left(\frac{s_1^2}{\pmb\chi+k_1}+\frac{s_2^2}{\pmb\chi+k_2}\right)=0,$
with $\eta=\eta_0+\epsilon^3$. We choose $\eta_0\in\mathbb{C}\backslash\mathbb{R}$ such that $\mathcal{P}(\pmb\chi;\eta_0)=(\pmb\chi-\chi_0)^3$. For a given $\eta_0$,
$\chi_0=\frac{2}{3}\eta_0-\frac{k_1+k_2}{3},s_1^2=\frac{(k_1+\chi_0)^3}{\eta_0(k_1-k_2)},s_2^2=\frac{(k_2+\chi_0)^3}{\eta_0(k_2-k_1)}$
and $k_i$ is taken in the form $k_i=\frac{2}{3}\im\eta_0\left[\cot(\omega_i)+\cot(\omega_1)+\cot(\omega_2)\right]-2\re\eta_0,$
with $\omega_1=\frac{\arg(\eta_0)}{3},\omega_2=\frac{\pi+\arg(\eta_0)}{3}$. For the purpose of precise analytical analysis, we carry out the coordinate transformation $(x,t)\mapsto (\pmb{x},\pmb{t})$ as
\small
\begin{align*}
\pmb{x}&=-x+(s_1^2+s_2^2+2\re\chi_0)t+\frac{1}{\im\chi_0}+\im\frac{\chi_2}{\chi_1^{2}}+C_1(\eta_0),\\
\pmb{t}&=2\im\chi_0t-\re\frac{\chi_2}{\chi_1^{2}}+C_2(\eta_0),
\end{align*}
\normalsize
where $C_1(\eta_0)$ and $C_2(\eta_0)$ are two constants to control the center of RW density located at $(0,0)$ on the $(\pmb{x}, \pmb{t})$-plane. In what follows, we focus on the FP-ESRW as an example, taking $\eta_0=1+\ii$ and $C_1(\eta_0)=-C_2(\eta_0)=-\frac{1}{4\im\chi_0}$. To investigate the phase properties of the vector RW, we analyze this solution within the framework of the topological vector potential described above.

\subsection{Simultaneous merger of two monopole pairs on the imaginary axis in the FPRW component}\label{SG1}

\begin{figure*}[htpb]	
	\centering
	\includegraphics[width=163mm]{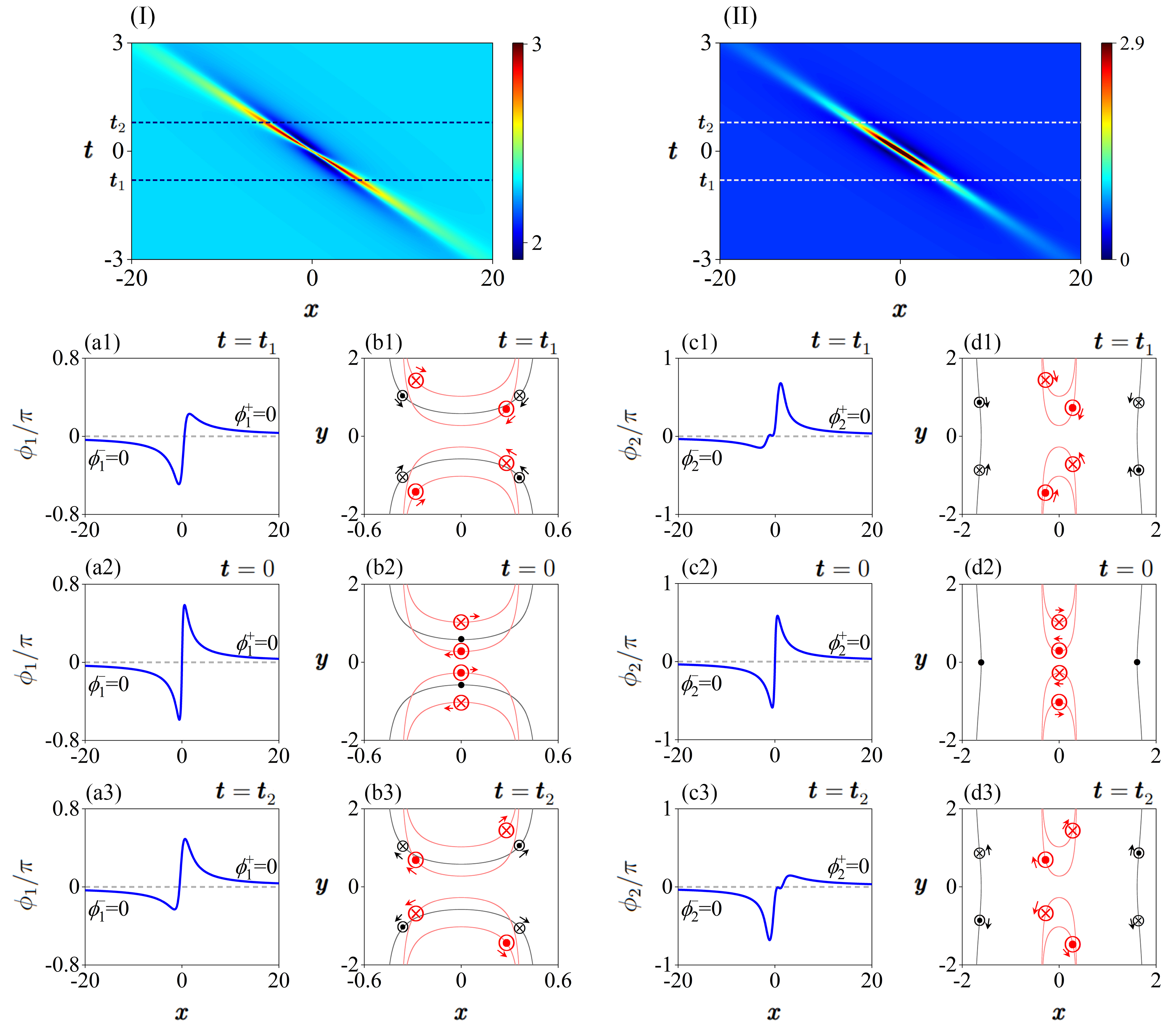}
	\caption{The topological phase of FP-ESRW. (I) and (II) represent the intensity distribution of FPRW and ESRW, respectively. Panels (a)-(b) and (c)-(d) depict the phase distribution and the associated magnetic field at selected times for the FPRW and ESRW, respectively. The trajectories of virtual monopoles with charges $\pm\frac{1}{2}$ (black curves) and $\pm\frac{3}{2}$ (red curves) are plotted in  (c) and (d), where the symbols $\bigodot$ and $\bigotimes$ denote positive and negative signs, respectively. Parameters are  $\pmb{t}_1=-\pmb{t}_2=-0.8.$}\label{fig2}
\end{figure*}
The density profile of component $q_1$ exhibits an FPRW pattern (Fig. \hyperref[fig2]{2(I)}), with the phase gradient flow given by $\phi_{1,\pmb{x}}=\frac{G_1(\pmb{x},\pmb{t})}{\prod_{i=1}^{4}[\pmb{x}-\pmb{x}_{1,+}^{[i]}(\pmb{t})][\pmb{x}-\pmb{x}_{1,-}^{[i]}(\pmb{t})]}$,
where $G_1(\pmb{x},\pmb{t})$ is provided in Appendix \ref{app2}, and the eight singularities are explicitly expressed as
$\pmb{x}_{1,\pm}^{[1]}=\pm\frac{\sqrt{m_1-\ii m_2}}{4\sqrt{2}}$, $\pmb{x}_{1,\pm}^{[2]}=\pm\frac{\sqrt{m_1+\ii m_2}}{4\sqrt{2}}$,
$\pmb{x}_{1,\pm}^{[3]}=\frac{1}{8}(\pm 3\ii +m_3)$, and $\pmb{x}_{1,\pm}^{[4]}=\frac{1}{8}(\pm 3\ii-m_3)$, with
$m_1=36-27\sqrt{3}-32\pmb{t}^2$, $m_2=24\sqrt{6}|\pmb{t}|\sqrt{2-\sqrt{3}}$ and $m_3=\sqrt{-27\pm 48\ii\pmb{t}-64\pmb{t}^2}$.
The density of the FPRW is $|q_1|^2=\frac{3}{2}s_1^2\frac{\prod_{i=1}^{2}[\pmb{x}-\pmb{x}_{1,+}^{[i]}(\pmb{t})][\pmb{x}-\pmb{x}_{1,-}^{[i]}(\pmb{t})]}
{\prod_{i=3}^{4}[\pmb{x}-\pmb{x}_{1,+}^{[i]}(\pmb{t})][\pmb{x}-\pmb{x}_{1,-}^{[i]}(\pmb{t})]}.$ Remarkably, the correspondence between density zeros/poles and phase singularities persists from the scalar case to the vector FPRW component. We label these eight singularities as $z_n = \pmb{x}_n + \ii \pmb{y}_n$ for $n = 1, \cdots, 8$, with $\pmb{x}_{1,-}^{[i]}$ and $\pmb{x}_{1,+}^{[i]}$ ($i = 1, \cdots, 4$) serving as the pair for $n = 2i-1$ and $n = 2i$, where $\pmb{x}_n = \mathrm{Re}[\pmb{x}_{1,\pm}^{[i]}]$ and $\pmb{y}_n = \mathrm{Im}[\pmb{x}_{1,\pm}^{[i]}]$.
In the topological magnetic field, these singularities appear as four pairs of virtual monopoles. The monopoles denoted $\pmb{x}_{1,\pm}^{[1]}$ and $\pmb{x}_{1,\pm}^{[2]}$ each carry a quantized flux $\pm\pi$, while those denoted $\pmb{x}_{1,\pm}^{[3]}$ and $\pmb{x}_{1,\pm}^{[4]}$ each carry a quantized flux $\pm3\pi$. Thus, the topological vector potential $\mathbf{A}_1$ of the FPRW component takes the same form as in Eq.~\eqref{TVP_sc_rw}.

In particular, $\pmb{x}_{1,\pm}^{[1]}$ and $\pmb{x}_{1,\pm}^{[2]}$ define two pairs of conjugate singularities in the complex plane, and the paired monopoles in the upper (or lower) half-plane share the same evolution path, as shown in Figs.~\hyperref[fig2]{2(b1)-(b3)} with black curves. Moreover, at the density center ($\pmb{t}=0$), we have $\pmb{x}_{1,\pm}^{[1]}=\pmb{x}_{1,\pm}^{[2]}$ and $\left(\prod_{k=1}^{2}[\pmb{x}-\pmb{x}_{1,+}^{[k]}(0)][\pmb{x}-\pmb{x}_{1,-}^{[k]}(0)]\right)\Big|G_1(\pmb{x},0)$, indicating that  the two corresponding pairs of monopoles simultaneously collide and merge in the upper and lower half-planes along the imaginary axis, respectively, as depicted by the black dots in Fig.~\hyperref[fig2]{2(b2)}. At this instant, the topological vector potential $\mathbf{A}_1$ is reduced to
{\fontsize{10pt}{10pt}\selectfont
\begin{align}\label{A11_t0}
	\mathbf{A}_1|_{\pmb{t}=0}=&+\frac{3}{2}\sum\limits_{n=5,8}\frac{
		(\pmb{x}-\pmb{x}_n)\mathbf{e}_{\pmb{y}}-(\pmb{y}-\pmb{y}_n)\mathbf{e}_{\pmb{x}}}{(\pmb{x}-\pmb{x}_n)^2+(\pmb{y}-\pmb{y}_n)^2}\nonumber\\
	&-\frac{3}{2}\sum\limits_{n=6,7}\frac{
		(\pmb{x}-\pmb{x}_n)\mathbf{e}_{\pmb{y}}-(\pmb{y}-\pmb{y}_n)\mathbf{e}_{\pmb{x}}}{(\pmb{x}-\pmb{x}_n)^2+(\pmb{y}-\pmb{y}_n)^2}.
\end{align}}

\noindent
When these two pairs of monopoles pass through $\pmb{t}=0$, their subsequent evolution is described by the transition rule $\pmb{x}_{1,-}^{[1]}\leftrightarrow \pmb{x}_{1,-}^{[2]}$, $\pmb{x}_{1,+}^{[2]}\leftrightarrow \pmb{x}_{1,+}^{[1]}$, which ensures the continuity of evolution for the virtual monopoles in the topological vector potential. In contrast, the two pairs of monopoles denoted $\pmb{x}_{1,\pm}^{[3]}$ and $\pmb{x}_{1,\pm}^{[4]}$ always evolve along four independent trajectories, as shown by the red curves in Figs. \hyperref[fig2]{2(b1)-(b3)}. Similar to the scalar case, a branch exchange is required for $\pmb{x}_{1,\pm}^{[3]}$ and $\pmb{x}_{1,\pm}^{[4]}$ as $\pmb{t}$ crosses 0 to maintain continuous evolution. Because monopoles in the same half-plane carry opposite charges, their contributions to the overall phase shift cancel exactly, yielding $\int_{-\infty}^{\infty}\textbf{A}_1(\pmb{x},\pmb{y}=0)\dd\pmb{x}=0$ and hence $\Delta\phi_1=0$ for all times.

To describe the topological feature of the FPRW, we present the asymptotic phase behavior of the FPRW at $\pmb{t}_1=-0.8$, $\pmb{t} = 0$, and $\pmb{t}_2 = 0.8$; phase variations are shown in Figs. \hyperref[fig2]{2(a1)-(a3)}. The asymptotic phases $\phi_1^{-}$ (as $\pmb{x}\to-\infty$) and $\phi_1^{+}$ (as $\pmb{x}\to+\infty$) satisfy $\phi_1^{-}-\phi_1^{+}=0$. The nontrivial phase variations are confined to regions where the density exhibits significant local structure. Figures \hyperref[fig2]{2(b1)-(b3)} present the corresponding trajectories and instantaneous positions of all eight virtual monopoles in the topological magnetic field, with arrows indicating their subsequent evolution directions. When the evolution reaches the density center at $\pmb{t} = 0$, two pairs of monopoles with charge $\pm 1/2$ collide on the imaginary axis, while two pairs of monopoles with charge $\pm 3/2$ simultaneously appear on the imaginary axis. This topological phase picture contrasts sharply with the scalar ESRW (Sec.~\ref{sec2}) and, more strikingly, with the upcoming vector ESRW component, where collisions occur on the real axis. This reveals the rich topological diversity of the vector DNLS system.

\subsection{Simultaneous merger of two monopole pairs on the real axis in the ESRW component}\label{SG2}

For component $q_2$, Eq.~\eqref{Vec_rwsol} produces an ESRW solution (Fig.~\hyperref[fig2]{2(II)}), but unlike the scalar case in Sec.~\ref{sec2}, its peak and two valleys coexist at $\pmb{t}=0$.
The vector ESRW density function is $|q_2|^2=\frac{3}{2}s_2^2\frac{\prod_{i=1}^{2}[\pmb{x}-\pmb{x}_{2,+}^{[i]}(\pmb{t})][\pmb{x}-\pmb{x}_{2,-}^{[i]}(\pmb{t})]}
{\prod_{i=3}^{4}[\pmb{x}-\pmb{x}_{2,+}^{[i]}(\pmb{t})][\pmb{x}-\pmb{x}_{2,-}^{[i]}(\pmb{t})]}.$ The corresponding phase gradient flow is given as $\phi_{2,\pmb{x}}=\frac{G_2(\pmb{x},\pmb{t})}{\prod_{i=1}^{4}[\pmb{x}-\pmb{x}_{2,+}^{[i]}(\pmb{t})][\pmb{x}-\pmb{x}_{2,-}^{[i]}(\pmb{t})]}$,
where $G_2(\pmb{x},\pmb{t})$ is given in Appendix \ref{app2}. Remarkably, the same correspondence between phase-gradient singularities and density zeros/poles observed in the $q_1$ component persists in the $q_2$ component. The eight singularities of $\phi_{2,\pmb{x}}$ are $\pmb{x}_{2,\pm}^{[1]}=\pm\frac{\sqrt{m_4-\ii m_5}}{4\sqrt{2}}$, $\pmb{x}_{2,\pm}^{[2]}=\pm\frac{\sqrt{m_4+\ii m_5}}{4\sqrt{2}}$, $\pmb{x}_{2,\pm}^{[3]}=\frac{1}{8}(\pm 3\ii +m_3)$, and $\pmb{x}_{2,\pm}^{[4]}=\frac{1}{8}(\pm 3\ii -m_3)$, with $m_4=36+27\sqrt{3}-32\pmb{t}^2$, $m_5=24\sqrt{6}|\pmb{t}|\sqrt{2+\sqrt{3}}$. These define complex coordinates $z_n = \pmb{x}_n + \ii \pmb{y}_n$ for $n = 1, \cdots, 8$  via the same pairing rule: for each $i=1,\cdots,4$, $\pmb{x}_{2,-}^{[i]}$ and $\pmb{x}_{2,+}^{[i]}$ correspond to $n=2i-1$ and $n=2i$, with $\pmb{x}_n = \mathrm{Re}[\pmb{x}_{2,\pm}^{[i]}]$, $\pmb{y}_n = \mathrm{Im}[\pmb{x}_{2,\pm}^{[i]}]$. They constitute four pairs of virtual monopoles within the topological vector potential $\mathbf{A}_2$, whose form is identical to that in Eq.~\eqref{TVP_sc_rw}.

We note that, unlike the FPRW, the paired monopoles described by the conjugate singularities $\pmb{x}_{2,\pm}^{[1]}$ ($\pmb{x}_{2,\pm}^{[2]}$) follow the same evolution path in the left (or right) half-plane (shown in Figs.~\hyperref[fig2]{2(d1)-(d3)} with black curves), while those from $\pmb{x}_{2,\pm}^{[3]}$ and $\pmb{x}_{2,\pm}^{[4]}$ evolve along four independent tracks (shown in Figs.~\hyperref[fig2]{2(d1)-(d3)} with red curves). Particularly, as $\pmb{t}\rightarrow 0$, we have $\pmb{x}_1=\pmb{x}_2=-\pmb{x}_3=-\pmb{x}_4$ and $\pmb{y}_1=-\pmb{y}_2\rightarrow 0, \pmb{y}_3=-\pmb{y}_4\rightarrow 0$. In this case, the topological vector potential $\mathbf{A}_2$ (Eq.~\eqref{TVP_sc_rw}) is reduced to
\begin{align*}
	\lim_{\pmb{t}\rightarrow\pm 0}\mathbf{A_2}=&\pm\pi\delta(\pmb{x}-\pmb{x}_1)\mathbf{e}_{\pmb{x}}\mp\pi\delta(\pmb{x}-\pmb{x}_3)\mathbf{e}_{\pmb{x}}\\
	&+\lim\limits_{\pmb{t}\rightarrow 0}\left[\frac{3\pmb{y}_{5}\mathbf{e}_{\pmb{x}}}{(\pmb{x}-\pmb{x}_{5})^2+\pmb{y}_{5}^2}-\frac{3\pmb{y}_{7}\mathbf{e}_{\pmb{x}}}{(\pmb{x}-\pmb{x}_{7})^2+\pmb{y}_{7}^2}\right],
\end{align*}
where the third term has no singularities in the limit, and its integral over $\mathbb{R}$ is $0$ since $\mathrm{sgn}(\pmb{y}_5)=\mathrm{sgn}(\pmb{y}_7)$. More remarkably, at $\pmb{t}=0$, we have $\pmb{x}_{2,\pm}^{[1]}=\pmb{x}_{2,\pm}^{[2]}=\pm\frac{\sqrt{36+27\sqrt{3}}}{4\sqrt{2}}$. This indicates that these two pairs of monopoles collide and merge symmetrically on the real axis. Because such a merger always produces opposite effects, it does not alter the final outcome or affect the phase shift $\Delta\phi_2$, as shown in Figs.~\hyperref[fig2]{2(c2)} and \hyperref[fig2]{2(d2)}. As a consequence, $\int_{-\infty}^{\infty}\textbf{A}_2(\pmb{x},\pmb{y}=0)\dd\pmb{x}=0$ and the phase shift $\Delta\phi_2=0$ for all times.

The phase characteristics of this vector ESRW are further illustrated in the right panel of Fig.~\hyperref[fig2]{2}. The phase distribution curves at $\pmb{t} = \pmb{t}_1 = -0.8$, $\pmb{t} = 0$, and $\pmb{t} = \pmb{t}_2 = 0.8$ are shown in Figs.~\hyperref[fig2]{2(c1)-(c3)}, respectively, where $\phi_2^\pm=\lim\limits_{\pmb{x}\rightarrow\pm\infty}\phi_2(\pmb{x})/\pi$. The trajectories and corresponding-time positions of all eight virtual monopoles are displayed in Figs.~\hyperref[fig2]{2(d1)-(d3)}, with arrows indicating their future evolution. In Fig.~\hyperref[fig2]{2(d2)}, the virtual monopoles with $\pm1/2$ charge disappear from the topological vector potential, as they merge on the real axis at $\pmb{t}=0$.

These results reveal a key difference between vector and scalar ESRWs (cf. Figs.~\hyperref[fig2]{2(d1)-(d3)} and Figs.~\hyperref[fig1]{1(c1)-(c5)}). In the scalar case, two $\pm1/2$ monopole pairs move asymmetrically: one pair merges and reemerges with opposite charge while the other remains unaffected, causing a phase shift jump until the second pair undergoes a similar merger. In the vector ESRW, the two $\pm1/2$ pairs move symmetrically and merge simultaneously with the appearance of the two valleys, so their contributions to $\Delta\phi_2$ cancel at all times. Thus, whether the phase shift changes depends on whether the monopole mergers on the real axis are staggered in time. Additionally, under other parameter choices in Eq.~\eqref{Vec_rwsol}, the vector system admits an anti-eye-shaped RW. This solution contains four virtual monopole pairs: two with charge $\pm1/2$ and two with $\pm3/2$, each following its own distinct trajectory. In contrast to the FPRW and the ESRW described above, the monopoles in the anti-eye-shaped RW never collide or merge (their trajectories are omitted here).

In previous studies \cite{Zhao5,Zhao2,Zhao1,Zhao4,Zhao3,Qin3}, phase singularities in nonlinear waves have yielded only the elementary monopole charge $\pm1/2$ in the complex plane. The present study alters this picture: as demonstrated above, both scalar and vector RWs in the DNLS system support a higher charge $\pm3/2$, absent in previous works. Beyond RWs, solitons offer a richer platform: whereas RWs contain only a finite number of monopole pairs, solitons support an infinite, periodically distributed array of virtual monopoles in the complex plane \cite{Zhao3}. This richer structure naturally raises the question of whether even higher charges can emerge. Our systematic study of BSs under the self-steepening effect reveals the emergence of higher quantized charges: they not only retain the $\pm3/2$ charge but also give rise to monopoles with a charge of $\pm5/2$, corresponding to third-order poles of the DHBS density in the vector system. In what follows, we investigate the topological phase properties of BSs in scalar and vector DNLS systems, revealing a hierarchy of monopole charges intricately linked to the order of zeros and poles of the DHBS density.

\section{$3/2$ topological charge in scalar bright solitons}\label{sec4}
\begin{figure}[!t]
\centering		
\includegraphics[width=80mm]{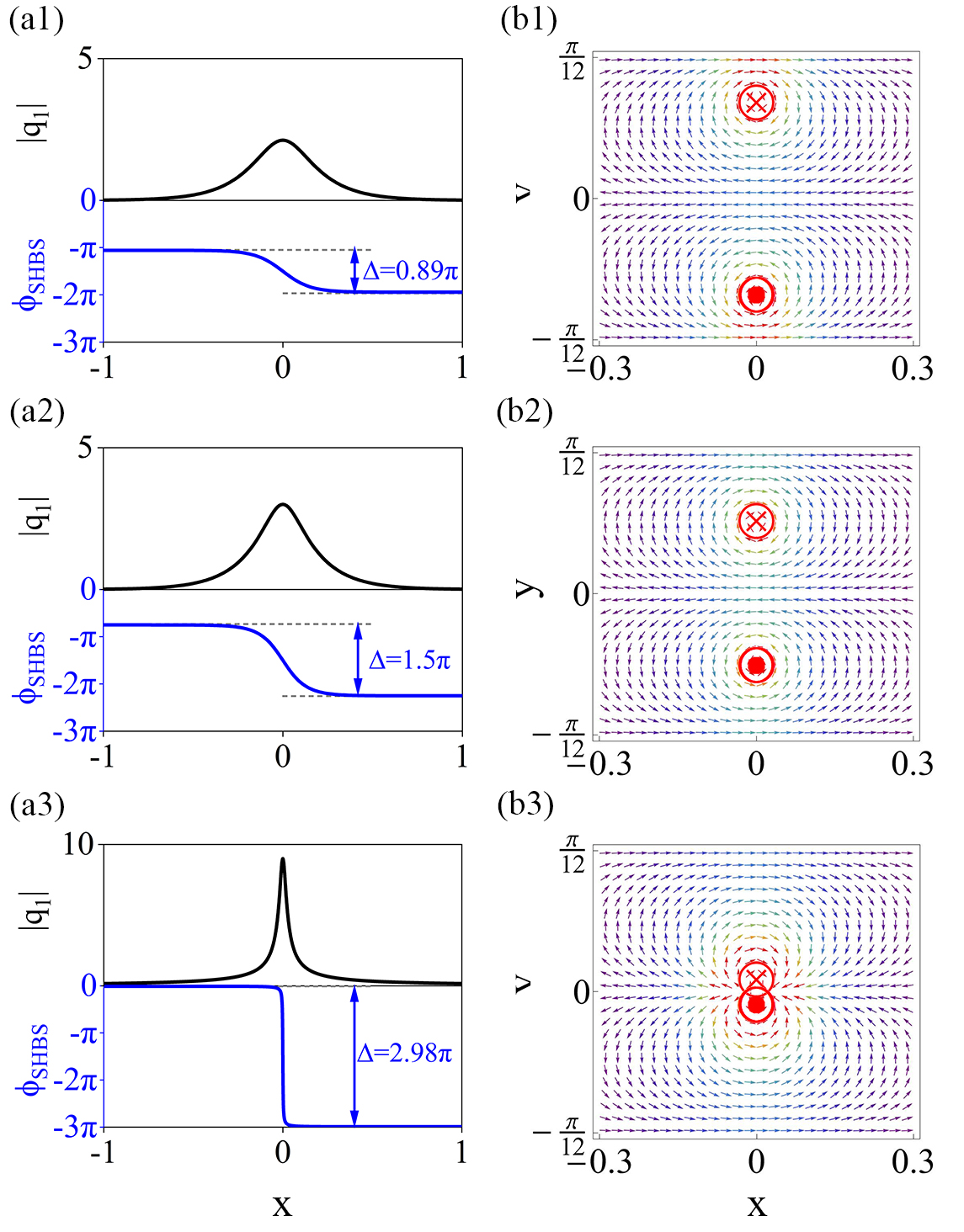}
\caption{Phase properties of scalar BSs. Panels (a1)-(a3) show the phase distributions (blue curve) alongside the intensity profiles (black curve). Panels (b1)-(b3) exhibit the corresponding topological vector potential, where monopoles with positive and negative charges are denoted by $\bigodot$ and $\bigotimes$, respectively. The parameters are set as: (1) $v_1=-9$; (2) $v_2=0$; (3) $v_1=599.94$. Other parameters are fixed as $\omega_1=6.$}\label{fig3}
\end{figure}

We first study the phase properties of scalar BSs. By performing the Darboux transformation \cite{Ling2}, the BS solution to the scalar  system \eqref{model} can be expressed as
\begin{align}\label{solution1}
q&=\frac{4\omega_1\left[a_1\cosh(\kappa_1)-\ii b_1\sinh(\kappa_1)\right]^3\e^{\ii R_1}}{\left[(a_1^2+ b_1^2)\cosh(2\kappa_1)+ a_1^2-b_1^2\right]^2},
\end{align}
with $\kappa_1=\omega_1(x-v_1t), \omega_1=6a_1b_1, v_1=6(b_1^2-a_1^2)$, and $R_1=\frac{v_1}{2}x+\left(\omega_1^2-\frac{v_1^2}{4}\right)t$. Here, $v_1$ and $\omega_1$ represent the velocity and width parameter of the BS, respectively. The density profile of this solution exhibits the typical SHBS structure. Then, we focus on  the topological phase properties by analyzing the exact SHBS solution.

The geometric phase of the SHBS can be defined as
$\phi\equiv\arctan\frac{f_2(f_2^2-3f_1^2)}{f_1(f_1^2-3f_2^2)}+n\pi$, with $f_1=a_1\cosh \kappa_1,f_2=b_1\sinh \kappa_1$. It satisfies
$\lim\limits_{x\rightarrow-\infty}\phi=-\lim\limits_{x\rightarrow+\infty}\phi$. Its gradient flow is
$\phi_x\!=\!-\frac{1}{2}|q|^2\!=\!-\frac{2\omega_1^2}{a_1^2\!+\!b_1^2}\frac{\e^{\eta}}{\e^{2\eta}\!+\!\mu\e^{\eta}\!+\!1}$,
where $\eta=2\omega_1 x, \mu=\frac{4}{1+u^2}-2, u=\frac{a_1}{b_1}$. This relation clearly reveals a direct correspondence between the singularities in the phase gradient and the density poles of the soliton. The phase jump is defined as $\Delta\phi\equiv-\int_{-\infty}^{\infty}\phi_x\dd x=3\left[\frac{\pi}{2}-\arctan\left(\frac{u^2-1}{2|u|}\right)\right]$, which yields $\Delta\phi \in (0, 3\pi)$. For example, we plot the phase variations of SHBSs along with their density distributions for different velocities in Figs.~\hyperref[fig3]{\ref{fig3}(a1)-\ref{fig3}(a3)} (with $\omega_1$ fixed at $6$), corresponding to velocities $v_1 = -9$, $v_1 = 0$, and $v_1 = 599.94$, respectively. Here, the black curve represents the density profile, while the blue curve shows the phase distribution. The phase jump increases with velocity. To visualize these phase shift variations more comprehensively, we plot their dependence on the $(v_1, \omega_1)$ plane in Fig.~\ref{fig4}, which reveals that $\Delta\phi$ is determined by the velocity-to-width ratio.

The scalar SHBS solution exhibits a striking feature in its phase shift, which spans an unprecedented range of $(0,3\pi)$ and attains $\Delta\phi = 3\pi/2$ in the stationary case---a value that may be particularly significant for topological characterization. This remarkable behavior sharply contrasts with known phase characteristics in other soliton systems. Dark solitons, for example, typically display their maximum phase shift in the static limit, with the shift decaying to zero at maximum velocity \cite{Kivshar,Kivshar1,Zhao3,Qin2}. Achieving the full $(0, 3\pi)$ range would require triple-valley dark solitons in Manakov systems, which exhibit a three-step structure in the phase jump \cite{Qin2}. The BSs in generalized scalar Chen-Lee-Liu systems exhibit phase shifts, which are strictly bounded by $\pi$ \cite{Zhao1}. In general, a one-step phase shift can give rise to a maximum phase shift of $\pi$, corresponding to the quantized flux of an elementary magnetic monopole with a charge $\pm1/2$ in the associated topological vector potential \cite{Kivshar,Kivshar1,Zhao3,Zhao1,Qin3}. Notably, the SHBS studied here exhibits a one-step structure with a $3\pi$ phase jump limit, suggesting that this SHBS possesses exceptional topological properties.

\begin{figure}[!t]
\centering		
\includegraphics[width=50mm]{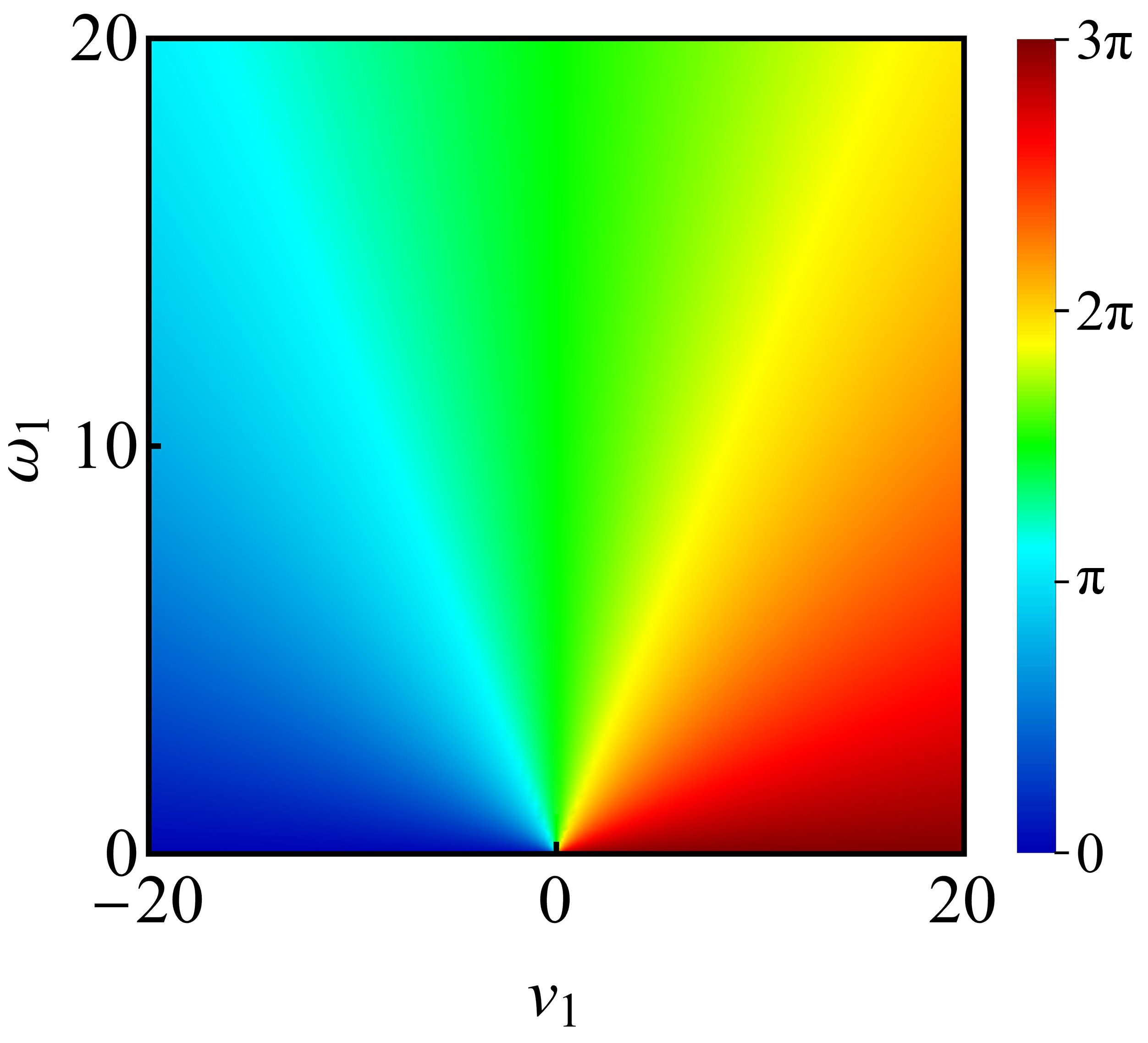}
\caption{The phase shift of a single-hump bright soliton in the $(v_1,w_1)$ plane at $t=0$.}\label{fig4}
\end{figure}

By extending the real coordinate $x$ to the complex variable $z = x + \ii y$, the phase gradient $\phi_x$ maps to its complex counterpart $\phi_z$. The singularities of $\phi_z$ are given by $z_n^\pm=\frac{1}{2\omega_1}\left[\pm\ii\arg(u^2-1+\ii 2|u|)+\ii2\pi n\right]$, where the $\pm$ signs are correlated and $n$ is an integer. These singularities correspond exactly to the positions of virtual magnetic monopoles carrying a quantized flux of $\pm3\pi$, i.e., charge $\sigma = \pm 3/2$ in the topological vector potential of the SHBS:
\begin{align}
\mathbf{A}=\sum_{n}\pm\frac{3}{2}\frac{(x-x_n)\mathbf{e}_y-(y-y_n)\mathbf{e}_x}{(x-x_n)^2+(y-y_n)^2}.
\end{align}
The period of the virtual monopole distribution is $T=\frac{\pi}{\omega_1}$. Each $z_n^\pm$ is a simple pole of both the phase gradient $\phi_z$ and the density function $|q|^2$ in the complex plane. This reveals that the associated monopole charges can attain the high value of $\pm \frac{3}{2}$. Such high charges are markedly different from the $\pm 1/2$ charges typical of previously studied solitons, such as dark solitons \cite{Zhao3,Qin3},  W-shaped solitons in the Sasa-Satsuma and Hirota models \cite{Zhao2}, and  BSs of the Chen-Lee-Liu equation \cite{Zhao1}. To illustrate this high charge dynamics,  Figs. \hyperref[fig3]{\ref{fig3}(b1)-(b3)} present the topological vector potentials of the SHBS under increasing velocity (fixed $w_1=6$), corresponding to the black curves in Figs. \hyperref[fig3]{\ref{fig3}(a1)-(a3)}. Over one period along the imaginary axis, virtual monopoles are distributed symmetrically about $y = 0$ within $[-\pi/12, \pi/12]$. As velocity increases, the decreasing monopole separation leads to a larger phase jump. The maximum $3\pi$ phase jump occurs precisely when monopoles collide and merge.

In particular, there are two monopoles described by $z_0^+$ and $z_0^-$ symmetrically distributed in a period range $[-\frac{T}{2},\frac{T}{2}]$ on the imaginary axis. Our analysis demonstrates that the phase shift can also be expressed in terms of the period, topological charge, and singularities, as
\begin{align}\label{phi2}
\Delta\phi\!=-\!\frac{2\pi}{T}\!\left[\!\sigma_{z_0^+}\!\left(\!\frac{T}{2}\!-\!\im z_0^+\right)\!+\!\sigma_{z_0^-}\!\left(\!-\frac{T}{2}\!-\!\im z_0^-\right)\!\right],
\end{align}
where $\sigma_{z_0^\pm}$ are the charges of the virtual monopoles corresponding to $z_0^\pm$. This formula reflects that the closer a virtual monopole approaches the real axis, the greater its contribution to the phase jump (Ref.~\cite{Zhao5}). The topological vector potential  gives $\Delta\phi=-\int_{-\infty}^\infty \mathbf{A}(x,y=0)\cdot\mathbf{e}_x\dd x$, so Eq.~\eqref{phi2} should be obtainable by integrating this potential. However, because the series  $\mathbf{A}(x,y=0)\cdot\mathbf{e}_x$ converges only conditionally to $\phi_x$ and its summation order lacks a well-defined sequence, we have only proven that under the Cauchy principal value of the periodic bracketed summation form, the integral converges to an expression of the form in Eq.~\eqref{phi2}. See Appendix~\ref{app0} for the proof.

\section{various topological charges in double-hump bright solitons}\label{sec5}

We recently demonstrated DHBSs in the vector DNLS system \cite{Qin}. The exotic topological phase properties of RWs and scalar BSs (revealed above) naturally raise the question of whether vector DHBSs exhibit richer topological structures. Our analysis confirms this expectation, revealing a hierarchy of quantized charges: $\pm1/2$, $\pm3/2$, and notably $\pm5/2$. The highest charge $\pm5/2$ corresponds to third-order poles of the DHBS density, a feature absent in the scalar case.

\subsection{Topological vector potentials of DHBSs with charges $\pm1/2$, $\pm3/2$, and $\pm5/2$} \label{SF}

The DHBS solution to the vector DNLS equation \eqref{model2} is obtained via the binary Darboux transformation \cite{Ling}. For the purposes of our analysis, its exact expressions are reformulated as follows:
\begin{subequations}\label{solution2}	
\begin{align}
q_{1}&=c_1\frac{V_{1}N^3}{H^2}\e^{\ii [\frac{v}{2}x+\left(\omega_1^2-\frac{v^2}{4}\right)t]},\\
q_{2}&=c_2\frac{V_{2}N^3}{H^2}\e^{\ii [\frac{v}{2}x+\left(\omega_2^2-\frac{v^2}{4}\right)t]},
\end{align}
\end{subequations}
where $v=-6\mathrm{Re}(\lambda_i^2)$ is the velocity, $\omega_i=3\mathrm{Im}(\lambda_i^2)$ are two width parameters, and $\lambda_i=a_1+\ii b_i$ is a spectral parameter. The exact expressions for $c_i$, $V_i$, $H$ and $N$ are given in Appendix \ref{app2}. This solution can generate either symmetric or asymmetric DHBS profiles by adjusting parameters. As examples, we show three types of symmetric DHBS profiles in Fig. \ref{fig5}: the first type in panels \hyperref[fig5]{\ref{fig5}(a1)} and \hyperref[fig5]{\ref{fig5}(a3)}, the second in \hyperref[fig5]{\ref{fig5}(b1)} and \hyperref[fig5]{\ref{fig5}(b3)}, and the third in \hyperref[fig5]{\ref{fig5}(c1)} and \hyperref[fig5]{\ref{fig5}(c3)}, all plotted with black curves. Then, we perform the  topological potential theory on DHBSs to uncover more interesting topological properties.

From the vector DHBS solution, the phase of each component is $\phi_i(x)= \arctan \frac{\im(V_i N^3)}{\re(V_i N^3)} + n\pi$ with $i = 1,2$ and integer $n$. The value of $n$ is chosen (based on the phase gradient direction) to ensure a continuous phase profile. This leads to the relation $\lim\limits_{x \to -\infty} \phi_i(x) = -\lim\limits_{x \to +\infty} \phi_i(x)$. Without loss of generality, we assume $\omega_1 > \omega_2 > 0$ and $t=0$ for convenience. Let $\omega = \frac{\omega_2}{\omega_1}$; then $\omega \in (0,1)$ and $\varpi = \frac{1-\omega}{1 + \omega} \in (0,1)$. The equation $\re(V_i N^3) = 0$ may possess roots, yet the phase gradient $F_i(x)=\frac{\dd \phi_i(x)}{\dd x}$  remains continuously extendable. Given the complexity of $\phi_i$, we first consider the static case ($v = 0$), where the phase gradient takes the compact form: $F_1(x)=-2\frac{\mathbf{\Lambda_1}\mathbf{R_1}}{S_1\mathbf{\Gamma}\mathbf{R}}\equiv\frac{\varLambda_1(x)}{\varGamma_1(x)}$,
$F_2(x)=-2\frac{\mathbf{\Lambda_2}\mathbf{R_2}}{S_2\mathbf{\Gamma}\mathbf{R}}\equiv\frac{\varLambda_2(x)}{\varGamma_2(x)}$. The explicit definitions of $S_i$, $\pmb{R}_i$, $\pmb{\varLambda}_i$, $\pmb{\Gamma}$ and $\pmb{R}$ are provided in Appendix \ref{app2}. Meanwhile, the density function for the DHBS is given by $|q_1(x)|^2=12|\beta_2|^2\omega_2\frac{S_1}{\mathbf{\Gamma}\mathbf{R}}$ and
$|q_2(x)|^2=12|\gamma_1|^2\omega_1\frac{S_2}{\mathbf{\Gamma}\mathbf{R}}$. Comparing $F_i(x)$ with $|q_i(x)|^2$ reveals that the singularities of the phase gradient correspond to the zeros of the density function from $S_i$, together with the poles from $\mathbf{\Gamma}\mathbf{R}$ and $S_i/(\mathbf{\Gamma}\mathbf{R})$. This links the zeros and poles of the density functions to the singularities of the vector potential in the complex plane, thus uncovering the topological properties that underlie the phase shift.

We now analytically continue the phase gradient $F_i(x)=\frac{\dd\phi_i(x)}{\dd x}$ and the density $|q_i(x)|^2$ to the complex plane by simply replacing $x$ with $z=x+\ii y$, yielding $F_i(z)$ and $|q_i(z)|^2$, respectively. Accordingly, the topological vector potential for a DHBS in two components is given by $\mathbf{A}_i=\mathrm{Re}[\frac{\partial\phi_i(z)}{\partial z}]\textbf{e}_x-\mathrm{Im}[\frac{\partial\phi_i(z)}{\partial z}]\textbf{e}_y$. Its general form can be expressed as
\begin{align}\label{A_i1}
\mathbf{A}_i=\sum_{n}\frac{\sigma_{i,n} [(x-x_{i,n})\mathbf{e}_y-(y-y_{i,n})\mathbf{e}_x]}{(x-x_{i,n})^2+(y-y_{i,n})^2},
\end{align}
where $\sigma_{i,n}$ represents the topological charge of the corresponding virtual monopole determined by the singularity $z_{i,n}=x_{i,n}+\ii y_{i,n}$. The distribution and charges of virtual monopoles in the topological vector potential $\mathbf{A}_i$ can be obtained by identifying the singularities of the complex phase gradient $F_i(z)$ in the complex plane, with the charges characterizing the topological properties of these singularities. Conversely, integrating $\mathbf{A}_i$ along the real $x$-axis reconstructs the phase variation.

Thereafter, we focus on the singularity property of the vector potential field $\mathbf{A}_i$ based on $F_i(z)$. Since the expression for $F_i(z)$ is in a fractional form, the analysis of its singularities needs to consider these two cases: (1) $\varGamma_i(z)=0$ and $\varLambda_i(z)\neq 0$. Singular points $z_{i,n}$ satisfying this condition are necessarily poles of $F_i(z)$. In particular, if $\partial_z\varGamma_{i}(z)\neq 0$ at $z=z_{i,n}$, then $z_{i,n}$ is a simple pole of $F_i(z)$; (2) $\varGamma_i(z)=\varLambda_i(z)=0$. To determine whether $z_{i,n}$ is a pole in this case, the orders of the zeros in the numerator and denominator must be analyzed separately. A pole exists if and only if the zero of the denominator is of higher order than that of the numerator. Then, we systematically classify the singularities of the topological vector potential by analyzing their charges.

Analysis of the explicit expressions for $F_1(z)$ and $F_2(z)$ reveals that $F_2(z)$ is a permuted version of $F_1(z)$ under the parameter substitution $(\beta_2,\gamma_1,\omega_1,\omega_2)\mapsto(\gamma_1,\beta_2,\omega_2,\omega_1)$. This structural isomorphism implies the two functions share analogous mathematical properties. Hence, it suffices to analyze $F_1(z)$, with the understanding that the singularities associated with $F_2(z)$ follow via the same isomorphic mapping. We begin with a mathematical analysis of  $\varGamma_1(z)=0$ for $F_1(z)$. Three cases must be considered.

\textbf{Case 1:}  $S_1=0$, monopoles carrying $\pm1/2$ charge.

Equation $\varGamma_1(z)=0$ naturally admits a family of solutions governed by the transcendental equation $\e^{2\omega_1z}\varpi^2+\e^{-2\omega_1z}|\gamma_1|^4=0$, which yields the general solution form $z_{1,n}^{\pm}=\frac{\ln|\chi|}{2\omega_1}\pm\frac{\pi\ii}{4\omega_1}+\frac{n\pi\ii}{\omega_1}$,
with $\chi=\ii|\gamma_1|^2/\varpi$ and $n\in\mathbb{Z}$. Under the conditions $\mathbf{\Gamma R} \neq \mathbf{0}$ and $\mathbf{\Lambda}_1 \mathbf{R}_1 \neq \mathbf{0}$, all these points are simple poles of the topological vector potential $\mathbf{A}_1$ with residues $\pm\frac{1}{2\ii}$, which implies that the corresponding magnetic monopoles of $\mathbf{A}_1$ have charges of $\pm\frac{1}{2}$.  Furthermore, these singularities correspond to the simple zeros of the density function $|q_1(z)|^2$.

\textbf{Case 2:}  $\mathbf{\Gamma}\mathbf{R}=0$, monopoles carrying  $\pm3/2$ charge.

The analysis becomes more complex when $\mathbf{\Gamma R} = 0$ under the condition that $e^{2\omega_1z}\varpi^2 + e^{-2\omega_1z}|\gamma_1|^4 \neq 0$. It is essential to determine whether each zero $z = z_{1,n}$ of $\mathbf{\Gamma R} = 0$ also satisfies $\varLambda_1(z) = 0$. If, instead, $\varLambda_1(z_{1,n}) \neq 0$ and $\partial_z \varGamma_1(z_{1,n}) \neq 0$, the corresponding virtual monopole charges can be directly evaluated using the residue formula: $\lim\limits_{z\rightarrow z_n}\frac{\varLambda_1(z)}{\partial_z\varGamma_{1}(z)}$. In this scenario, the resulting gradient flow singularities correspond to simple poles of the density function. Our extensive numerical calculations confirm that the virtual monopoles in this case carry quantized charges of $\pm \frac{3}{2}$. Several illustrative examples will be presented later.

\textbf{Case 3:} $S_1\!\!=\!\!\mathbf{\Gamma}\mathbf{R}\!\!=\!\!0$, monopoles carrying  $\pm5/2$ charge.
	
The most challenging scenario arises when both $S_1=e^{2\omega_1z}\varpi^2 + e^{-2\omega_1z}|\gamma_1|^4 = 0$ and $\mathbf{\Gamma R} = 0$ are simultaneously satisfied. Here, we analyze this case under the symmetric condition given by $2\omega_1\ln{|\beta_2|}-2\omega_2\ln|\gamma_1|=(\omega_1-\omega_2)\ln \left|\frac{\omega_1-\omega_2}{\omega_1+\omega_2}\right|$, with $\omega = \frac{4n_1+1}{4n_2+3} \in (0,1)$, where $n_1, n_2 \in \mathbb{Z}^+ \cup \{0\}$, to facilitate analytical treatment. This choice of $\omega = \frac{4n_1+1}{4n_2+3} \in (0,1)$ is detailed in the appendix \ref{app1}. The case where $\omega = \frac{4n_1+3}{4n_2+1} \in (0,1)$ with $n_1, n_2 \in \mathbb{Z}^+ \cup \{0\}$ can be discussed similarly. The corresponding zeros are then calculated as $z_{1,n}^\pm=\frac{\ln|\chi|}{2\omega_1}\pm\frac{(4n_2+3)\pi\ii}{4\omega_1}+\frac{ n_3(4n_2+3)\pi\ii}{\omega_1}$, where $n_3\in\mathbb{Z}$. However, these solutions $z_{1,n}^\pm$ simultaneously satisfy the equation $\mathbf{\Lambda_1 R_1} = 0$. As pointed out above, the orders of the zeros in the numerator and denominator of $F_i(z)$ must be analyzed separately to determine whether $z_n$ constitutes a pole in this case. For a fixed ratio $\omega$, the coordinate transformations $Z = \omega_1 z$ and $\omega_2 z = \omega Z$ show that the solutions are scaling invariant with respect to $\omega_1$ and $\omega_2$. This allows us to normalize $2\omega_1 = 1$ to simplify the analysis without sacrificing generality. Easily, we have $\frac{\partial^m\mathbf{\Gamma}\mathbf{R}}{\partial z^m}=\mathbf{\Gamma}J_0^m\mathbf{R},~\frac{\partial^m\mathbf{\Lambda_1}\mathbf{R_1}}{\partial z^m}=\mathbf{\Lambda_1}J_1^m\mathbf{R_1}$, where
$J_0=\diag\left(1+\omega,-1-\omega,1-\omega,-1+\omega,0\right)$, $J_1=\diag\left(1+\omega,-1-\omega,1-\omega,-1+\omega,2,-2,0\right)$.
Through direct calculation, we find that for $m_0,m_1\le 3$, both $\mathbf{\Gamma}J_0^{m_0}\mathbf{R}|_{z=z^\pm}$ and $\mathbf{\Lambda_1}J_1^{m_1}\mathbf{R_1}|_{z=z^\pm}$ vanish. However, when $m_0=m_1=4$, these expressions become non-zero. This demonstrates that the critical points correspond to fourth-order zeros in the numerator and fifth-order zeros in the denominator of the gradient flow $F_1(z) = \frac{\partial \phi_1}{\partial z}$, respectively. As a result, these critical points behave as simple poles. The residues at these singularities are given by $\lim\limits_{z\rightarrow z^{\pm}}(z-z^\pm)\frac{\partial \phi_1}{\partial z}=\pm\frac{5}{2\ii}$, which reveals that the magnetic monopole charges associated with $\mathbf{A}_1$ at these poles are $\pm \frac{5}{2}$. They also coincide with fourth-order zeros of $\mathbf{\Gamma}\mathbf{R}$ and simple zeros of $S_1$, indicating that they are in fact third-order poles of the density function $|q_1(z)|^2$. These poles differ from those of multi-pole Kaup-Newell solitons \cite{PS1,PS2}, which originate from spectral-plane scattering-data multiplicities. In the extended complex plane,  high odd-multiple charges require the appropriate density-pole structure, but spectral multiplicity alone is insufficient.

Combining the three cases, virtual monopoles within the topological vector potential of vector DHBSs can admit charges of $\pm1/2$, $\pm3/2$, and $\pm5/2$, though not all necessarily simultaneously. Their appearance depends on specific conditions. For the $i$-th component, we denote the numbers of monopoles with charges $\pm1/2$, $\pm3/2$, and $\pm5/2$ by $r_{i,1}$, $r_{i,2}$ and $r_{i,3}$, respectively. Then, the topological vector potential in Eq.~\eqref{A_i1} can be explicitly expressed in the periodic symmetric summation form as:
\begin{align}\label{VA}	\mathbf{A}_i\!=\!&\sum_{n=-\infty}^{\infty}\!\bigg[\!\left(\!+\frac{1}{2}\right)\!\sum_{j=1}^{r_{i,1}}\frac{(x\!-\!x_{i,j})\mathbf{e}_y\!-\![y\!-\!(y_{i,j}^+\!+\!nT)]\mathbf{e}_x}{(x\!-\!x_{i,j})^2\!+\![y\!-\!(y_{i,j}^+\!+\!nT)]^2}\nonumber\\
&\!+\!\left(\!-\frac{1}{2}\right)\!\sum_{j=1}^{r_{i,1}}\frac{(x\!-\!x_{i,j})\mathbf{e}_y\!-\![y\!-\!(y_{i,j}^-\!+\!nT)]\mathbf{e}_x}{(x\!-\!x_{i,j})^2\!+\![y\!-\!(y_{i,j}^-\!+\!nT)]^2}\nonumber\\
&\!+\!\left(\!+\frac{3}{2}\right)\!\sum_{j=1}^{r_{i,2}}\frac{ (x\!-\!x_{i,j})\mathbf{e}_y\!-\![y\!-\!(y_{i,j}^+\!+\!nT)]\mathbf{e}_x}{(x\!-\!x_{i,j})^2\!+\![y\!-\!(y_{i,j}^+\!+\!nT)]^2}\nonumber\\
&\!+\!\left(\!-\frac{3}{2}\right)\!\sum_{j=1}^{r_{i,2}}\frac{ (x\!-\!x_{i,j})\mathbf{e}_y\!-\![y\!-\!(y_{i,j}^-\!+\!nT)]\mathbf{e}_x}{(x\!-\!x_{i,j})^2\!+\![y\!-\!(y_{i,j}^-\!+\!nT)]^2}\nonumber\\
&\!+\!\left(\!+\frac{5}{2}\right)\!\sum_{j=1}^{r_{i,3}}\frac{ (x\!-\!x_{i,j})\mathbf{e}_y\!-\![y\!-\!(y_{i,j}^+\!+\!nT)]\mathbf{e}_x}{(x\!-\!x_{i,j})^2\!+\![y\!-\!(y_{i,j}^+\!+\!nT)]^2}\nonumber\\
&\!+\!\left(\!-\frac{5}{2}\right)\!\sum_{j=1}^{r_{i,3}}\frac{ (x\!-\!x_{i,j})\mathbf{e}_y\!-\![y\!-\!(y_{i,j}^-\!+\!nT)]\mathbf{e}_x}{(x\!-\!x_{i,j})^2\!+\![y\!-\!(y_{i,j}^-\!+\!nT)]^2}\bigg],
\end{align}
where $T$ is the period of the magnetic monopole distribution. In addition, $y_{i,j}^+=-y_{i,j}^-$, $|y_{i,j}^\pm|<\frac{T}{2}$ and the sign ``$\pm$'' of $y_{i,j}^\pm$  corresponds to the sign of the charge. In particular, the phase jumps of the static solitons can be quantitatively determined through line integrals of the corresponding topological vector potential along the real axis: $\Delta \phi_1= -\int_{-\infty}^{\infty}\mathbf{A}_1(x,y=0)\cdot \mathbf{e}_x\dd x=\frac{7\pi}{2},\Delta \phi_2= -\int_{-\infty}^{\infty}\mathbf{A}_2(x,y=0)\cdot \mathbf{e}_x\dd x=\frac{5\pi}{2}$. When generalized to arbitrary velocity regimes ($v\neq 0$),  the phase shift of these vector DHBSs can vary within the range $(0, 6\pi)$.

\subsection{Examples of virtual monopoles with various charges for DHBSs}

\begin{figure*}[htpb]	
\centering
\includegraphics[width=175mm]{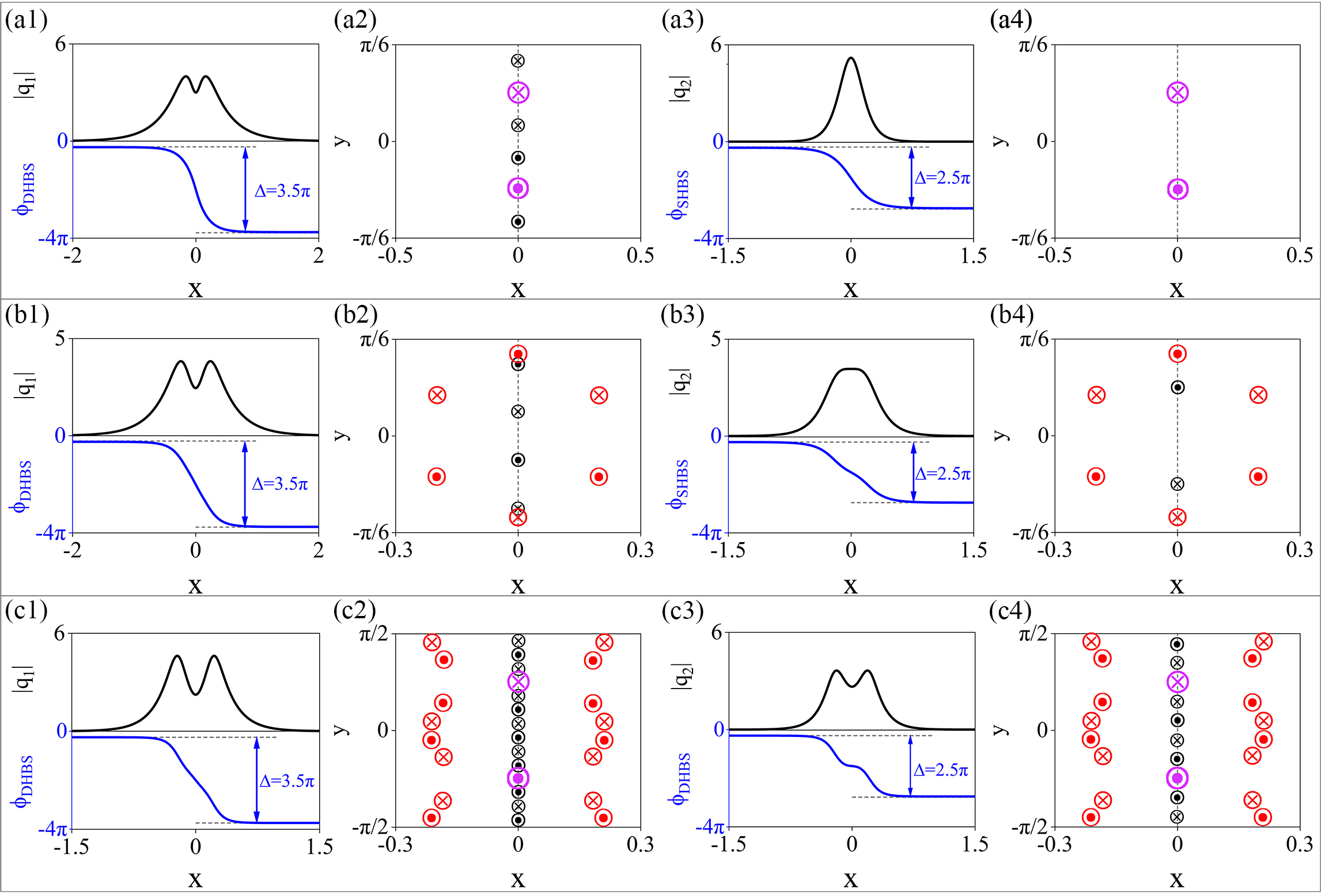}
\caption{Abundant topological charges of the virtual magnetic fields for three types of  vector symmetric DHBSs.  Rows (a1)-(a4), (b1)-(b4), and (c1)-(c4) correspond to the three distinct cases with parameters given below. In each row, the black curves in the top panels represent the density profiles. The corresponding phase variations in real space are plotted in blue below them, and the virtual magnetic fields in the complex plane are shown to the right. The virtual magnetic monopoles are color- and symbol-coded: black, red, and purple for charges of $\pm\frac{1}{2}$, $\pm\frac{3}{2}$, and $\pm\frac{5}{2}$; and $\bigodot$ and $\bigotimes$ for positive and negative signs, respectively. These virtual monopoles are shown over one period along the $y$ axis. Parameters for each case are as follows: (a) $\omega_1=3\omega_2=9,$ $\beta_2=\gamma_1=\frac{1}{\sqrt{2}}$;  (b) $\omega_1=2\omega_2=6,~\beta_2=\gamma_1=\frac{1}{\sqrt{3}}$; (c) $\omega_1=7,\omega_2=5,\beta_2=\gamma_1=\frac{1}{\sqrt{6}}$.}\label{fig5}
\end{figure*}

We now present several examples to demonstrate the rich topological phase characteristics of DHBSs by conducting the above analysis. We begin by considering the specific case of $\omega = \frac{1}{3}$. Combining this with the symmetric condition, the resulting gradient flows in the two components are $F_1(x)=\frac{\omega_1-4\omega_1\cosh\left[\frac{2}{3}(2\omega_1x-\ln|\chi|)\right]}{\cosh(2\omega_1x-\ln|\chi|)}$ and
$F_2(x)=-\frac{5\omega_2}{\cosh(2\omega_2x-\ln|\chi|)}$, with $|\chi|=2|\gamma_1|^2=8|\beta_2|^6$. The singularities of the corresponding analytic extensions $F_1(z)$ and $F_2(z)$ can then be accurately calculated from these expressions. For example, by setting $|\chi|=1$ and $|\gamma_1|^2=|\beta_2|^2=\varpi=\frac{1}{2}$, the solution Eq.~\eqref{solution2} gives a symmetric DHBS-SHBS pattern with axis $x=0$, and the density profiles of the two components have been shown in Figs. \hyperref[fig5]{\ref{fig5}(a1)} and \hyperref[fig5]{\ref{fig5}(a3)}. The phase variation in both components is confirmed to exhibit a distinct one-step structure, as demonstrated by the blue curve below the density profiles. This dramatic change in the phase jump structure can be attributed to the underlying topological potentials. Under analytic extensions, the distribution of the corresponding virtual magnetic field within one period is depicted in Figs. \hyperref[fig5]{\ref{fig5}(a2)} and \hyperref[fig5]{\ref{fig5}(a4)}, respectively. In this case, the singularities of $F_i(z)$ are simple poles obtained by \textbf{Case 1} and \textbf{Case 3}. The virtual monopoles located at $\left(0,\pm\frac{\pi}{36}\right)$ and $\left(0,\pm\frac{5\pi}{36}\right)$, marked by small black symbols, carry a charge of $\mp\frac{1}{2}$. However, this type of monopole is confined to component $q_1$. Moreover, unlike the virtual monopole fields reported in Refs. \cite{Zhao1,Qin3,Zhao3,Zhao4,Zhao5}, these magnetic monopoles with positive and negative charges are not alternately distributed along the imaginary axis but are symmetrically distributed about the real axis. Strikingly, virtual monopoles with charge $\mp\frac{5}{2}$ emerge at $(0,\pm\frac{\pi}{12})$ in both components, and are indicated by big purple symbols. The above analysis reveals that the point-like magnetic field in Fig. \hyperref[fig5]{\ref{fig5}(a2)} originates from the simple zeros and triple poles of the DHBS density function in component $q_1$. In contrast, the field in Fig. \hyperref[fig5]{\ref{fig5}(a4)} is solely due to the triple poles of the SHBS density function in component $q_2$.

Then, we consider the case for virtual monopoles with charges of $\pm \frac{1}{2}$ and $\pm\frac{3}{2}$ by setting $\omega = \frac{1}{2}$ and $|\beta_2|^2 = |\gamma_1|^2 = |\varpi| = \frac{1}{3}$. With this parameter choice, the solution Eq.~\eqref{solution2} yields another symmetric DHBS-SHBS pattern, as depicted in   Figs.~\hyperref[fig5]{\ref{fig5}(b1)} and \hyperref[fig5]{\ref{fig5}(b3)} with black curves. The density profile of the SHBS in component $q_2$ possesses a triple pole at $x=0$, which is distinct from the structures shown in Fig.~\hyperref[fig5]{\ref{fig5}(a3)}. In this case, the equation $\mathbf{\Gamma R} = 0$ reduces to a cubic equation in $\cosh(Z)$: $2\cosh(Z)^3+3\cosh(Z)+4=0$. Introducing $\zeta = 3\sqrt{2} - 4$, the solutions for $\cosh(Z)$ are given by $\cosh(Z)=\frac{-2^{2/3}+2^{1/3}\zeta^{2/3}}{2\zeta^{1/3}},\frac{-2^{2/3}(-1\pm\ii\sqrt{3})+2^{1/3}(-1\mp\ii\sqrt{3})\zeta^{2/3}}{4\zeta^{1/3}}$, corresponding to \textbf{Case 2}. Then, the singularities $z_n$ can be obtained analytically and correspond to the simple poles of the density function. Analytical computation of the residues $\mathbf{Res}[F_i(z_n)]$ confirms that these virtual monopoles carry a quantized magnetic flux of $3\pi$ and a charge of $\pm\frac{3}{2}$, as indicated by the red symbols in Fig. \hyperref[fig5]{5(b2)} and Fig.~\hyperref[fig5]{5(b4)}. These virtual monopoles in both components share identical coordinates. Within one period, they are located approximately at $(0,\pm0.441)$, $(0.198,\pm 0.220)$, and $(-0.198,\pm 0.220)$. By further performing \textbf{Case 1}, we get two pairs of virtual monopoles with charge $\mp\frac{1}{2}$ for the DHBS shown in Fig. \hyperref[fig5]{5(b2)} with black symbols, which are located at $\left(0,\pm\frac{\pi}{24}+\frac{n\pi}{6}~(n=0,1)\right)$, while there is one pair of monopoles with charge $\mp\frac{1}{2}$ for the SHBS in Fig. \hyperref[fig5]{5(b4)}, located at $\left(0,\pm\frac{\pi}{12}\right)$. Such a virtual magnetic monopole distribution yields a one-step structural phase variation along the real axis for the DHBS in component $q_1$, but a two-step variation for the SHBS in component $q_2$, interestingly, as depicted in Figs.~\hyperref[fig5]{\ref{fig5}(b1)} and \hyperref[fig5]{\ref{fig5}(b3)} with blue curves, in contrast to the phase variations for the symmetric DHBS-SHBS shown in Figs. \hyperref[fig5]{\ref{fig5}(a1)} and \hyperref[fig5]{\ref{fig5}(a3)}.

We further investigate a more complex symmetric configuration, with parameters chosen as $\omega_1 = 7$, $\omega_2 = 5$, and $\beta_2 = \gamma_1 = \frac{1}{\sqrt{6}}$. The corresponding density profiles exhibit a symmetric DHBS-DHBS pattern, as shown by the black curves in Figs. \hyperref[fig5]{\ref{fig5}(c1)} and \hyperref[fig5]{\ref{fig5}(c3)}. Surprisingly, by examining the singularities of the phase gradient, we find that they arise from simple zeros, simple poles, and third-order poles of the density function, respectively. Consequently, the topological phase features in this case involve topological charges of $\pm1/2$, $\pm3/2$, and $\pm5/2$, corresponding to \textbf{Case 1} (simple zeros), \textbf{Case 2} (simple poles), and \textbf{Case 3} (third-order poles).  Under the topological theory framework presented above, the virtual magnetic monopole fields generated by these singularities can be obtained, as illustrated in Figs. \hyperref[fig5]{\ref{fig5}(c2)} and \hyperref[fig5]{\ref{fig5}(c4)}.
Here, we also exhibit the distribution of monopoles along the $y$ axis over a period. It can be seen that the distributions of the virtual monopoles are much more abundant than those in the two examples shown in Figs. \hyperref[fig5]{\ref{fig5}(a)} and \hyperref[fig5]{\ref{fig5}(b)}. The virtual monopoles with charge $\mp\frac{1}{2}$, marked by black points and obtained by solving \textbf{Case 1}, are located at $\left(0,\ \pm\frac{\pi}{28} \pm \frac{n\pi}{7}\right)$ for $-3 \le n \le 3,\ n \ne -2$ in Fig. \hyperref[fig5]{5(c2)}, and at $\left(0,\ \pm\frac{\pi}{20} \pm \frac{n\pi}{5}\right)$ for $-2 \le n \le 2,\ n \ne 1$ in Fig. \hyperref[fig5]{5(c4)}. It is evident that these virtual monopoles with positive and negative charges alternate along the imaginary axis. However, their distributions differ between the two components, resulting in different contributions to the phase distribution. The virtual monopoles marked by red symbols correspond to those obtained from solving \textbf{Case 2}. They share identical coordinates in both components, which can be approximately given by $\left(-0.183, \pm 0.425\mp\frac{n\pi}{2}\right)$ and $\left(0.183, \pm 0.425\mp\frac{n\pi}{2}\right)$ with charge $\pm(-1)^n\frac{3}{2}$, and $\left(-0.211, \pm 0.141\mp\frac{n\pi}{2}\right)$ and $\left(0.211, \pm 0.141\mp\frac{n\pi}{2}\right)$ with charge $\mp(-1)^n\frac{3}{2}$ ($n = 0, 1$). The purple magnetic monopoles carry charge $\mp\frac{5}{2}$ and are located at $\left(0,\ \pm\frac{\pi}{4}\right)$ in both components, as obtained by solving \textbf{Case 3}. Remarkably, these virtual monopole fields cause the DHBS phase jump to exhibit a three-step structure in component $q_1$ and a two-step structure in component $q_2$, as shown by the blue curves in Figs. \hyperref[fig5]{5(c1)} and \hyperref[fig5]{5(c3)}. Although the triple-step structure of the DHBS is not readily apparent in Fig. \hyperref[fig5]{5(c1)}, it can be confirmed by the topological vector potential. Comparing the three cases in Fig.~\ref{fig5} reveals distinct virtual monopole fields behind seemingly similar soliton and phase profiles, making the underlying topological phase structure a useful criterion for distinguishing solitons.

These results reveal the abundant and exotic topological phase characteristics of vector DHBSs in a coupled nonlinear system featuring self-steepening but lacking the conventional cubic nonlinearity. Notably, the system exhibits a topological vector potential that supports the coexistence of monopoles with charges $\pm1/2$ and, more remarkably, high odd multiple magnetic charges $\pm3/2$, and $\pm5/2$. Through a direct analytical correspondence, these charges are traced back to simple zeros, simple poles, and third-order poles of the density function, respectively, thereby establishing a new relation between higher-order density singularities and high odd multiple topological charges.

\subsection{Collision-induced reorganization of DHBS topological vector potentials}

Consistent with Dirac's monopole theory, this correspondence retains its full validity throughout the collision of two DHBSs. As reported in Ref. \cite{Qin}, such collisions in the vector DNLS system are generally inelastic. This inelasticity necessarily entails a restructuring of the density function's analytic structure in the complex plane before and after the collision. Consequently, the virtual monopole fields undergo a corresponding redistribution in position, number, and charge assignment, while the quantized higher-order odd-multiple charge values themselves, dictated solely by the order of density singularities, remain intrinsically invariant. We illustrate this point using an exact collision between two DHBSs.

The exact two-DHBS solution is obtained through a four-fold Darboux transformation and can be expressed in the following determinant form:
\begin{align}\label{eq4}
q_1^{[4]}=\sqrt{\frac{2}{3}}
\left(\frac{\det M_1}{\det M}\right)_x,
q_2^{[4]}=\sqrt{\frac{2}{3}}
\left(\frac{\det M_2}{\det M}\right)_x.
\end{align}
Here, the block matrix $M$ is
$M=
\begin{pmatrix}
G_{1,4}&\cdots&G_{1,1}\\
\vdots&\ddots&\vdots\\
G_{4,4}&\cdots&G_{4,1}
\end{pmatrix},$ with
$G_{i,j}=
\begin{pmatrix}
\lambda_i^{2j}\e^{\theta_i}&\lambda_i^{2j-1}\beta_i&\lambda_i^{2j-1}\gamma_i\\
-\lambda_i^{*2j}\beta_i^*&\lambda_i^{*2j-1}\e^{\theta_i^*}&0\\
-\lambda_i^{*2j}\gamma_i^*&0&\lambda_i^{*2j-1}\e^{\theta_i^*}
\end{pmatrix}$.
The matrices $M_1$ and $M_2$ are obtained from $M$ by replacing its eleventh and twelfth columns, respectively, by $E=\left(-\e^{\theta_1},\beta_1^*,\gamma_1^*,\ldots,
-\e^{\theta_4},\beta_4^*,\gamma_4^*\right)^{\T}$, with $\theta_i=\omega_i(x-v_it)+\ii R_i, R_i=\frac{v_i}{2}x+\left(\omega_i^2-\frac{v_i^2}{4}\right)t, v_i=-6\operatorname{Re}(\lambda_i^2), \omega_i=3\operatorname{Im}(\lambda_i^2).$ The two-DHBS solution Eq.~\eqref{eq4} provides an exact description of the full collision dynamics. To determine how the interaction affects the topological phase of each DHBS, we derive asymptotic expressions from this solution for the two DHBSs before ($S_1^{-}$ and $S_2^{-}$, as $t\rightarrow-\infty$) and after ($S_1^{+}$ and $S_2^{+}$, as $t\rightarrow+\infty$) the collision based on Eq.~\eqref{eq4}. The explicit asymptotic expressions take the following common form:
\begin{subequations}\label{AS}
\begin{align}
q_1^{S_k^{\pm}}
&=\sqrt{\frac{2}{3}}
\left[
\frac{c_{1k}^{\pm}
\left(\zeta_{1k}^{\pm}\e^{\xi_r}
+\eta_{1k}^{\pm}\e^{-\xi_r}\right)\e^{\ii R_s}}
{\mathcal D_k^{\pm}}
\right]_x, \\
q_2^{S_k^{\pm}}
&=\sqrt{\frac{2}{3}}
\left[
\frac{c_{2k}^{\pm}
\left(\zeta_{2k}^{\pm}\e^{\xi_s}
+\eta_{2k}^{\pm}\e^{-\xi_s}\right)\e^{\ii R_r}}
{\mathcal D_k^{\pm}}
\right]_x,
\end{align}
\end{subequations}
where
$
\mathcal D_k^{\pm}
=d_{1k}^{\pm}\e^{-\nu_{rs}^{+}}
+d_{2k}^{\pm}\e^{\nu_{rs}^{+}}
+d_{3k}^{\pm}\e^{-\nu_{rs}^{-}}
+d_{4k}^{\pm}\e^{\nu_{rs}^{-}}$, $\nu_{rs}^{+}=\xi_r+\xi_s$,
$\nu_{rs}^{-}=\xi_r-\xi_s$  and $\xi_j=\omega_j(x-v_jt)$.
To display the coefficients explicitly, define $\Delta_{jk}=\lambda_j^2-\lambda_k^2, \Delta_{j\bar{k}}=\lambda_j^2-\lambda_k^{*2}$.
The active spectral pair is $(r,s)=(3,4)$ for $k=1$ and $(r,s)=(1,2)$ for $k=2$. The other relevant expressions are given in Appendix \ref{app:collision-solution}.
These asymptotic forms allow the topological vector potential of each DHBS to be determined separately.  With explicit analytical expressions available for the density and phase gradient of every branch, we can perform analytic continuation to the extended complex plane and locate the zeros and poles together with their orders. Each branch then possesses a topological vector potential of the form given in Eq.~\eqref{VA}, and the corresponding virtual monopoles and their quantized charges follow from the zeros and pole orders according to the correspondence established in the preceding sections.

\begin{figure*}[!t]
	\centering
	\includegraphics[width=\textwidth]{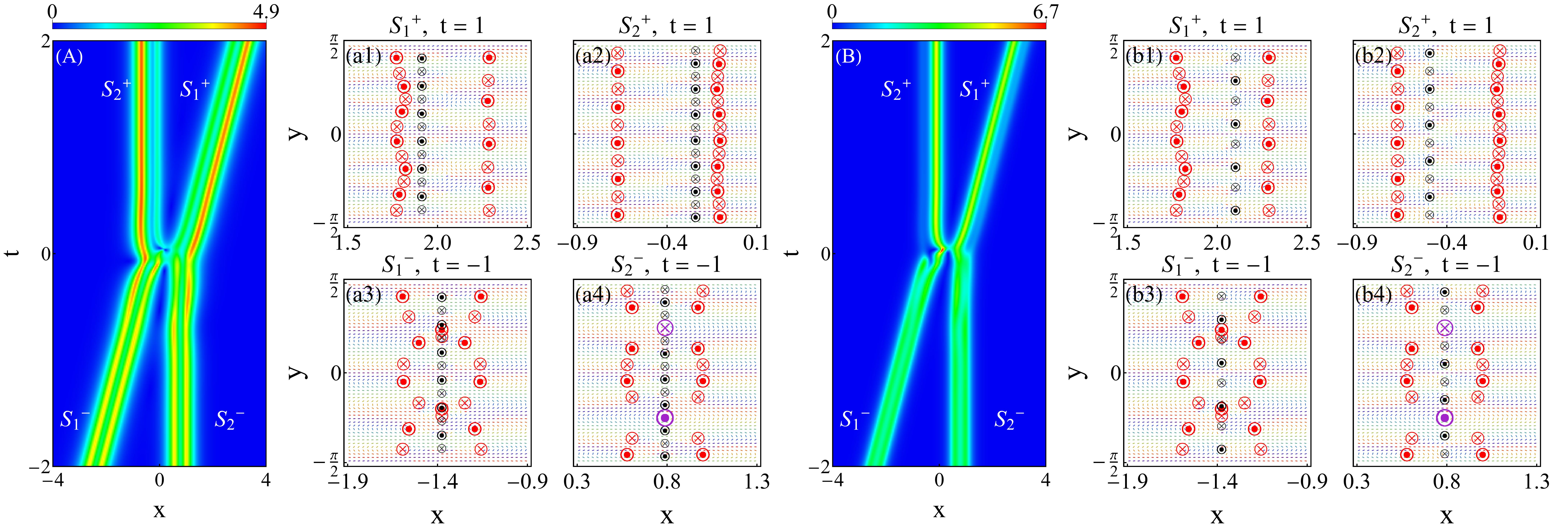}
	\caption{Inelastic collision between two DHBSs and their asymptotic topological vector fields. Panels (A) and (B) show the density evolutions of the two DHBSs' collision, labeled $S_1$ and $S_2$, in the first and second components, respectively. Panels (a) and (b) display the asymptotic topological vector fields, respectively. Panels (a1)-(a2) and (b1)-(b2) show the post-collision branches $S_1^+$ and $S_2^+$ at $t=1$, whereas panels (a3)-(a4) and (b3)-(b4) show the pre-collision branches $S_1^-$ and $S_2^-$ at $t=-1$. The coordinates of each magnetic monopole are shown in Appendix \ref{coordinates}.  The inelastic collisions are observed to redistribute the positions, numbers, and charges of virtual monopoles in each DHBS's topological vector potential fields. The parameters are $\lambda_1=\sqrt{7/6}(1+\ii)$, $\lambda_2=\sqrt{5/6}(1+\ii)$, $\lambda_3=1+13\ii/12$, $\lambda_4=25/32+85\ii/96$, $\gamma_1=\sqrt{2}$, $\beta_2\approx0.316$, $\gamma_3=1$, and $\beta_4\approx1.078$, with $\beta_1=\gamma_2=\beta_3=\gamma_4=0$.}
	\label{fig:6}
\end{figure*}

As a concrete illustration, we consider the collision of two DHBSs with the parameters given in the caption of Fig.~\ref{fig:6}. Panels \hyperref[fig:6]{6(A)} and \hyperref[fig:6]{6(B)} show the density evolutions of the interacting solitons in the first and second components, respectively. Notably, two initially symmetric DHBSs become asymmetric after the collision. The associated topological vector potentials are presented in panels  \hyperref[fig:6]{6(a)} and \hyperref[fig:6]{6(b)} for the two components, respectively: panels \hyperref[fig:6]{6(a1)}-\hyperref[fig:6]{6(a2)} and \hyperref[fig:6]{6(b1)}-\hyperref[fig:6]{6(b2)} show the branches  $S_1^+$ and $S_2^+$ at $t=1$, while panels \hyperref[fig:6]{6(a3)}-\hyperref[fig:6]{6(a4)} and \hyperref[fig:6]{6(b3)}-\hyperref[fig:6]{6(b4)} display $S_1^-$ and $S_2^-$ at $t=-1$. Before the collision, the virtual magnetic fields of each DHBS are symmetrically distributed about the $x$ axis; after the collision, this symmetry is broken. The moving DHBS ($S_1$) exhibits a periodic distribution of virtual monopoles along the $y$ axis with period $512\pi$, whereas the static one ($S_2$) has period $\pi$. For direct comparison, all panels are shown in the same $y$-window $[-\pi/2,\pi/2]$, with the corresponding monopole coordinates and charges listed in Appendix~\ref{coordinates}.

Strikingly, although both DHBSs are symmetric before the collision, the moving one does not host $\pm5/2$  monopoles at any stage. The static one, by contrast, possesses  $\pm5/2$ monopoles before the collision but loses them afterward, retaining only $\pm1/2$ and $\pm3/2$ charges. Moreover, the numbers and positions of the monopoles also undergo significant rearrangement in both components. These observations collectively demonstrate that the inelastic collision redistributes the positions, numbers, and charge assignments of virtual monopoles within each DHBS's topological vector potential fields. This suggests that collision properties can be distinguished through their topological phases, rather than from intensity profiles alone.

\section{Conclusion And Discussion}\label{sec6}

In summary, by applying Dirac's magnetic monopole theory in an extended complex plane, we have uncovered high topological charges encoded in the phases of nonlinear waves with self-steepening but without the usual cubic nonlinearity. For RWs, both scalar and vector types carry the same four pairs of virtual monopoles, two of charge $\pm1/2$ arising from simple zeros (as reported previously for certain nonlinear waves) and two of charge $\pm3/2$ arising from simple poles. Although the monopole content is identical, its dynamics differ markedly. In the scalar ESRW, the two $\pm1/2$ pairs merge at temporally separated events, resulting in a $2\pi$ phase shift. In the vector ESRW, however, their symmetric simultaneous merging cancels the phase shift entirely. For the vector FPRW, merging occurs on the imaginary axis and likewise does not affect the phase shift. We further extend our analysis to BSs. For the scalar SHBS, the density simple poles constitute periodic point-like virtual monopoles of charge $\pm3/2$. Remarkably, the third-order poles of vector DHBSs support an even higher charge of $\pm5/2$, in addition to the $\pm1/2$ and $\pm3/2$ charges. These findings not only confirm a series of quantized magnetic charges for virtual monopoles underlying nonlinear waves, but also establish a correspondence whereby the higher charges $\pm3/2$ and $\pm5/2$ arise from simple poles and third-order poles of the density function, respectively. The zeros and poles both play crucial roles for topological charges, since both of them can induce phase singularities. This is somewhat similar to the case that topological invariant defined in terms of a complex function, counting its zeros and poles inside the unit circle for giving the full classification of critical phases \cite{Verresen}.  Furthermore, using the inelastic collision of two DHBSs, we demonstrate that the density-charge correspondence still holds in this process, with the inelastic collision redistributing the virtual monopoles in position, number, and charge assignment while the quantized charges remain intact.

These results uncover that higher charges such as $\pm3/2$ and $\pm5/2$ exist at poles of the density functions, at which phase singularities are induced, further revealing general characteristics of Dirac monopoles underlying the nonlinear waves. This strongly suggests that the general Dirac monopole theory still holds in the extended complex plane. Accordingly, the virtual magnetic fields can be used to reasonably address phase singularity problems. This greatly motivates the systematic exploration of quantized charges encoded in the nonlinear waves. It is suggested that high odd-multiple topological charges can be obtained from high-order odd poles of the density functions of nonlinear waves in the complex plane. Generally, a density pole of odd order $2m-1$ admits a virtual monopole charge $\pm(2m+1)/2$. Analogously, high even-multiple charges are expected to arise from high-order even poles. The recent observation of $\pm1$ topological charges through density zeros of the Akhmediev breather in real spatial-temporal space further suggests that high even charges could be generated by multiple zeros of wave functions in real space \cite{Zhao4}.  On the other hand, nonlinear waves governed by both self-steepening effects and usual cubic nonlinearities host $\pm1/2$ charges \cite{Zhao1}, similar to the purely cubic case but distinct from the pure self-steepening scenario discussed above, indicating that the pole or zero characteristics could be related to the model type. Nonlinear wave solutions of different models could admit identical charges, but some models can admit different poles with specific charges. The underlying mechanism of how the poles of density functions generate topological singularities remains unknown, as do the intrinsic relations between physical terms in the models and the zero or pole characteristics. Therefore, the close relations among zeros or poles of wave function density, topological charges, and different physical effects still require further investigation.

\section*{ACKNOWLEDGMENTS}
Y.-H. Q. is supported by the National Natural Science Foundation of China (Grant No. 12405004), the Natural Science Foundation of Xinjiang Uygur Autonomous Region Project (Grant No. 2024D01C232), the Scientific Research Projects Funded by the Basic Research Business Expenses of Autonomous Region Universities (Grant No. XJEDU2024P011), the ``Tianchi Talent'' Introduction Plan in Xinjiang Uygur Autonomous Region, and the Tianshan Innovation Team Program of Xinjiang Uygur Autonomous Region (Grant No. 2026D14003). L.-C. Z. is supported by the National Natural Science Foundation of China (Contracts No. 12375005 and No. 12235007), and the Major Basic Research Program of Natural Science of Shaanxi Province (Grant No. 2018KJXX-094).
	
\appendix
\section{The related exact expressions}\label{app2}

The function $F(\pmb{x},\pmb{t})$ in the phase gradient $\phi_{\pmb{x}}(\pmb{x},\pmb{t})$ of the scalar RW in Sec. \ref{sec2} is given by
\begin{align*}
	F(\pmb{x},\pmb{t})=\frac{8}{J_{[2,1]}^+}\sum_{j=0}^{6}F_j(s,\mu;\pmb{x})\pmb{t}^j,
\end{align*}
where $J_{[m,n]}^\pm=s^m( \mu^2\pm 1)^n$, and
\begin{align*}
	F_0&\!=\!\frac{-6+J^+_{[4,1]}\pmb{x}^2\left(13-4s^4\pmb{x}^2+J^+_{[8,2]}\pmb{x}^4\right)}{J^+_{[12,3]}},\\
   F_1&\!=\!\frac{\pmb{x}\!\left[25\mu^2\!-\!1\!-\!\left(8J_{[4,1]}^-\!+\!10J_{[4,2]}^+\right)\pmb{x}^2\!+\!J_{[0,1]}^-J_{[8,2]}^+\pmb{x}^4\right]}{\mu^2J_{[8,2]}^+},\\
	F_2&\!=\!\frac{-13-12\left(J_{[4,1]}^-+\mu^2J_{[4,1]}^+\right)\pmb{x}^2+J_{[8,2]}^+\pmb{x}^4}{\mu^2J_{[8,2]}^+},
\end{align*}
\begin{align*}	
F_3&\!=\!\frac{2\pmb{x}(-1+12\mu^2+7\mu^4+2\mu^6+J_{[0,1]}^-J_{[4,2]}^+\pmb{x}^2)}{\mu^4J_{[4,2]}^+},\\
	F_4&\!=\!\frac{-4\!-\!4J_{[0,2]}^+\!-\!J_{[4,2]}^+\pmb{x}^2}{\mu^4J_{[4,2]}^+},F_5\!=\!\frac{J_{[0,1]}^-\pmb{x}}{\mu^6},F_6\!=\!\frac{-1}{\mu^6}.
\end{align*}

The function $G_1(\pmb{x},\pmb{t})$ in the phase gradient $\phi_{1,\pmb{x}}$ of the FPRW in Sec. \ref{SG1} is given by
 \begin{align*}
 	G_1(\pmb{x},\pmb{t})=\frac{3}{2^{17}}\sum_{j=0}^{6}G_{1}^{[j]}(\pmb{x})\pmb{t}^j,
 \end{align*}
 where
 \begin{align*}
 	G_{1}^{[0]}\!=&3^7(43-24\sqrt{3})-2^53^6(17-10\sqrt{3})\pmb{x}^2\\
 	   &+2^{10}3^4(3-2\sqrt{3})\pmb{x}^4-2^{15}\cdot3\pmb{x}^6,\\
 	G_{1}^{[1]}\!=&2^6\pmb{x}\big[3^4(71\!\sqrt{3}\!-\!126)\!+\!2^63^2(18\!-\!11\!\sqrt{3})\pmb{x}^2\!-\!2^{10}\!\sqrt{3}\pmb{x}^4\big],\\
 	G_{1}^{[2]}\!=&2^5\big[3^5(59\!-\!30\sqrt{3})\!-\!2^63^2(21\!-\!14\sqrt{3})\pmb{x}^2\!-\!2^{15}\cdot 3\pmb{x}^4\big],\\
 	G_{1}^{[3]}\!=&-2^{12}\pmb{x}\left[3^2(18-7\sqrt{3})+2^5\sqrt{3}\pmb{x}^2\right],G_{1}^{[5]}=2^{16}\sqrt{3}\pmb{x},\\
 	G_{1}^{[4]}\!=&2^{10}\cdot 3\left[3^2(17-6\sqrt{3})+2^5\pmb{x}^2\right],G_1^{[6]}=3\times 2^{15}.
 \end{align*}

The function $G_2(\pmb{x},\pmb{t})$ in the phase gradient $\phi_{2,\pmb{x}}$ of the ESRW in Sec. \ref{SG2} is given by
\begin{align*}
 	G_2(\pmb{x},\pmb{t})=\frac{3}{2^{17}}\sum_{j=0}^{6}G_{2}^{[j]}(\pmb{x})\pmb{t}^j,
\end{align*}
where
\begin{align*}
 	G_{2}^{[0]}\!=&3^7(43+24\sqrt{3})-2^53^6(17+10\sqrt{3})\pmb{x}^2\\
 	&+2^{10}3^4(3+2\sqrt{3})\pmb{x}^4-2^{15}\cdot3\pmb{x}^6,\\
 	G_{2}^{[1]}\!=&2^6\!\pmb{x}\!\big[\!-\!3^4\!(71\!\sqrt{3}\!+\!126)\!+\!2^63^2\!(18\!+\!11\!\sqrt{3})\pmb{x}^2\!+\!2^{10}\!\sqrt{3}\pmb{x}^4\big],\\
 	G_{2}^{[2]}\!=&2^5\big[3^5(59\!+\!30\sqrt{3})\!-\!2^63^2(21\!+\!14\sqrt{3})\pmb{x}^2\!-\!2^{15}\cdot 3\pmb{x}^4\big],\\
 	G_{2}^{[3]}\!=&-2^{12}\pmb{x}\left[3^2(18+7\sqrt{3})-2^5\sqrt{3}\pmb{x}^2\right],G_{2}^{[5]}\!=\!-2^{16}\sqrt{3}\pmb{x},\\
 	G_{2}^{[4]}\!=&2^{10}\cdot 3\left[3^2(17+6\sqrt{3})+2^5\pmb{x}^2\right],G_2^{[6]}\!=3\times 2^{15}.
\end{align*}

The exact expressions for the  vector DHBS solution Eq.~\eqref{solution2} are given by
\begin{align*}	
 &H \!=\!|\varrho_{1}\e^{-\kappa_+}\!+\!\varrho_{2}\e^{\kappa_+}\!+\!\varrho_{3}\e^{-\kappa_-}\!+\!\varrho_{4}\e^{\kappa_-}|^2,\\
 &\varrho_1\!=\!{|\beta_2\gamma_1|^2}\sqrt{(v^2+4\omega_1^2)(v^2+4\omega_2^2)},\\
 &\varrho_2\!=\!{\varpi^2}\left(v-2\ii\omega_1\right)\left(v-2\ii\omega_2\right),\\
 &\varrho_3\!=\!{|\gamma_1|^2}(-v+2\ii\omega_2)\sqrt{v^2+4\omega_1^2},\\
&\varrho_4\!=\!{|\beta_2|^2}(-v+2\ii\omega_1)\sqrt{v^2+4\omega_2^2},\\
&V_1\!=\!\e^{-\omega_1x-\ii\theta_1}|\gamma_1|^2\!-\e^{\omega_1x+\ii\theta_1}\varpi\!,\\
&V_2\!=\!\e^{-\omega_2x-\ii\theta_2}|\beta_2|^2+\!\e^{\omega_2x+\ii\theta_2}\varpi\!,
\end{align*}
\begin{align*}	
&N\!=\!\e^{-\kappa_++\ii(\theta_1+\theta_2)}|\beta_2\gamma_1|^2+\e^{\kappa_+-\ii(\theta_1+\theta_2)}\varpi^2\\
 	&~~~~~+\e^{-\kappa_-+\ii(\theta_1-\theta_2)}|\gamma_1|^2+\e^{\kappa_--\ii(\theta_1-\theta_2)}|\beta_2|^2,\\
 	&\theta_1\!=\!\arctan(b_1/a_1),\theta_2\!=\!\arctan(b_2/a_2),\\
 	&\varpi\!=\!\frac{\omega_1-\omega_2}{\omega_1+\omega_2},\kappa_\pm\!=\!(\omega_1\pm\omega_2)(x-vt),	\\
 	&c_1\!=\!2\sqrt{6}\beta_2^* \omega_2(v^2+4\omega_1^2)(v^2+4\omega_2^2)^\frac{3}{4}\e^{-2\ii\theta_1},\\
 	&c_2\!=\!2\sqrt{6}\gamma_1^* \omega_1(v^2+4\omega_2^2)(v^2+4\omega_1^2)^\frac{3}{4}\e^{-2\ii\theta_2}.
 \end{align*}

The exact expressions for the phase gradient $F_i(z)$ of the vector DHBS in Sec. \ref{SF} are given by
\begin{align*}
 	&S_1=\e^{2\omega_1x}\varpi^2+\e^{-2\omega_1x}|\gamma_1|^4,\\
 	&S_2=\e^{2\omega_2x}\varpi^2+\e^{-2\omega_2x}|\beta_2|^4,\\	
 	&\mathbf{R}=\left(\e^{2\kappa_+},\e^{-2\kappa_+},\e^{2\kappa_-},\e^{-2\kappa_-},1\right)^T,\\
 	&\mathbf{\Lambda_{11}}\!=\!2|\gamma_1|^2\omega_1\frac{2+\omega}{1+\omega}(\varpi^4,|\beta_2\gamma_1|^4),\\
 	&\mathbf{\Lambda_{21}}\!=\!2|\beta_2|^2\omega_2\frac{1+2\omega}{1+\omega}(\varpi^4,|\beta_2\gamma_1|^4),\\
 	&\mathbf{\Lambda_{12}}\!=\!2|\gamma_1|^2\omega_1\varpi\frac{2-\omega}{1+\omega}(|\beta_2|^4,|\gamma_1|^4),\\
 	&\mathbf{\Lambda_{22}}\!=\!2|\beta_2|^2\omega_2\varpi\frac{1-2\omega}{1+\omega}(|\beta_2|^4,|\gamma_1|^4),\\
 	&\mathbf{\Lambda_{13}}\!=\!3|\beta_2|^2\omega_2\left(\varpi^4,|\gamma_1|^8,\frac{2|\gamma_1|^4\varpi}{3}\frac{7-3\omega^2}{(1+\omega)^2}\right),\\
 	&\mathbf{\Lambda_{23}}\!=\!3|\gamma_1|^2\omega_1\left(\varpi^4,|\beta_2|^8,\frac{2|\beta_2|^4\varpi}{3}\frac{3-7\omega^2}{(1+\omega)^2}\right),\\
 	&\mathbf{\Lambda_1}\!=\!\left(\mathbf{\Lambda_{11}},\mathbf{\Lambda_{12}},\mathbf{\Lambda_{13}}\right),
 	~\mathbf{\Lambda_2}\!=\!\left(\mathbf{\Lambda_{21}},\mathbf{\Lambda_{22}},\mathbf{\Lambda_{23}}\right),\\
 	&\mathbf{\Gamma}=\left(\varpi^4,|\beta_2\gamma_1|^4,|\beta_2|^4,|\gamma_1|^4,\frac{8\omega|\beta_2\gamma_1|^2}{(\omega+1)^2}\right),\\	&\mathbf{R}_1=\left(\e^{2\kappa_+},\e^{-2\kappa_+},\e^{2\kappa_-},\e^{-2\kappa_-},\e^{4\omega_1x},\e^{-4\omega_1x},1\right)^T,\\
 	&\mathbf{R}_2=\left(\e^{2\kappa_+},\e^{-2\kappa_+},\e^{2\kappa_-},\e^{-2\kappa_-},\e^{4\omega_2x},\e^{-4\omega_2x},1\right)^T.
 \end{align*}

\section{Topological vector potential and phase shift of the nonlinear localized waves}\label{app0}

We begin by discussing the properties of a meromorphic function, which satisfies the following theorem:

\begin{theorem} \label{the1}
	Consider a meromorphic function $M(z)$ with $m$ singularities $z_n ~(1\le n\le m)$, each having a corresponding multiplicity $l_n$. Then,
	$$M(z)=F(z)+\sum_{n=1}^{m}f_n,~~f_n=\sum_{j=1}^{l_n}\frac{A_{n,j}}{(z-z_n)^{j}},$$
	where $F(z)$ is a holomorphic function,  and
	$$A_{n,j}=\frac{1}{(l_n-j)!}\lim\limits_{z\rightarrow z_n}\frac{\dd^{l_n-j}}{\dd z^{l_n-j}}[(z-z_n)^{l_n}M(z)].$$
\end{theorem}
\begin{proof}
	Expanding the function $M(z)$ into a Laurent series at $z=z_n$, and denoting the holomorphic part by $F_n$, we obtain $M(z)=F_n(z)+f_n(z).$
	Consider the function $W_n(z)=(z-z_n)^{l_n}M(z)$.
	Then
	$$W_n(z)=(z-z_n)^{l_n}F_n(z)+\sum_{j=1}^{l_n}A_{n,j}(z-z_n)^{l_n-j}.$$
	On the other hand, expand $W_n$ into a Taylor series at $z=z_n$; then
	$W_n(z)=\sum_{j=0}^{\infty}w_j(z-z_n)^j,$
	where $w_j=\frac{1}{j!}\partial_z^{j}W_n|_{z=z_n}$.
	Thus $A_{n,l_n-j}=w_j$, i.e.,
	$$A_{n,j}=w_{l_n-j}=\frac{1}{(l_n-j)!}\partial_z^{l_n-j}W_n|_{z=z_n}.$$
	For another singularity $z=z_k$, setting  $F_n+f_n=F_k+f_k+f_n,$ where $F_k$ is holomorphic at $z=z_k$, similarly we can obtain $f_k$.  In this way, we can separate out all $f_n$, leaving the remaining part as the holomorphic function $F(z)$.
\end{proof}
Based on Theorem \ref{the1}, if $l_n=1$ and $F(z)=0$, setting $\mathbf{Res}[M,z_n]=\Omega_n/(2\pi \ii)$, one obtains
$$M(z)=\sum_{n=1}^{m}\frac{-\ii\Omega_n}{2\pi(z-z_n)}=\sum_{n=1}^{m}\frac{-\ii\Omega_n(z^*-z_n^*)}{2\pi|z-z_n|^2}.$$
Consider the vector function $M=\phi_x(z)\mathbf{e}_x+\phi_y(z)\mathbf{e}_y$, where $z=x+y\ii$. Let $\mathbf{Res}[\phi_x,z_n]=\Omega_{n,1}/(2\pi \ii)$ and $\mathbf{Res}[\phi_y,z_n]=\Omega_{n,2}/(2\pi \ii)$. Since $\ii \phi_x(z)=\phi_y(z)$, we have $\ii\Omega_{n,1}=\Omega_{n,2}$ and hence $ M(z)=\sum_{n=1}^{m}\frac{\Omega_{n,1}(\mathbf{e}_x+\ii\mathbf{e}_y)[-\ii(x-x_n)-(y-y_n)]}{2\pi[(x-x_n)^2+(y-y_n)^2]}.$
Hence
\begin{align} \label{A1}
	\mathbf{A}\!=\!\re[M(z)]\!=\!\sum_{n=1}^{m}\frac{\sigma_n[(x-x_n)\mathbf{e}_y\!-\!(y-y_n)\mathbf{e}_x]}{[(x-x_n)^2+(y-y_n)^2]},
\end{align}
where $\sigma_n=\Omega_{n,1}/(2\pi)$. For RWs with a finite number of singularities, expression \eqref{A1} can accurately characterize their topological vector potential. For solitons, which possess countably many singularities, the expression for the topological vector potential corresponds to the limit of \eqref{A1} as $m\rightarrow\infty$. However, as $m\rightarrow+\infty,$ $\mathbf{A}(x,y=0)\cdot\mathbf{e}_x$ is not absolutely convergent, and hence it does not admit a strict summation order. It remains unclear how the series should be summed in a specific order to converge to the phase gradient flow. Since Formula \eqref{phi2} is directly provided by the phase gradient flow, it offers a reference. We aim to investigate a summation form whose integral converges to an expression with the structure of \eqref{phi2}.

Suppose a soliton has $2p$ singularities $z_n^\pm=x_n+\ii y_n^\pm$ in a magnetic monopole distribution period $[-\frac{T}{2},\frac{T}{2}]$, and $y_n^+=-y_n^-, |y_n^\pm|<\frac{T}{2}$. Their charges are $\sigma_n^\pm$, and $\sigma_n^+=-\sigma_n^-.$ Then $x_n^\pm+\ii (y_n^\pm+ NT)$ are also singularities with charges $\sigma_n^\pm$.  Here the ``$\pm$'' sign in $y_n^\pm$ corresponds to positive and negative values, while the ``$\pm$'' in $\sigma_n^\pm$ is associated with the sign of
$y_n^\pm$ and does not indicate the sign of the charge. We assume that \eqref{A1} has the form
\begin{align}\label{AAA}
	\textbf{A}\!=\!\sum_{N=-\infty}^{\infty}&\sum_{n=1}^{p}\Bigg(\sigma_n^+\frac{(x\!-\!x_n)\mathbf{e}_y\!-\![y\!-\!(y_n^+\!+\!NT)]\mathbf{e}_x}{|z-z_{n,N}^+|^2}\nonumber\\
	&\!+\!\sigma_n^-\frac{(x\!-\!x_n)\mathbf{e}_y\!-\![y\!-\!(y_n^-\!+\!NT)]\mathbf{e}_x}{|z-z_{n,N}^-|^2}\Bigg).
\end{align}
Note that the period of the singularity distribution differs from that of the magnetic monopole distribution. Let $I_0=\left[-\frac{I}{2},\frac{I}{2}\right]$ be the period interval of the singularity distribution. If no singularities lie on the boundary of $I_0$, we can take $T=I$; if singularities do lie on the boundary, their conjugacy property implies the existence of magnetic monopole pairs with opposite topological charges on the upper and lower boundaries. Hence, $I_0$ is not a period interval for the magnetic monopole distribution. Since there are no magnetic monopoles on the real axis, we may take $T=2I$ in this case. Under the Cauchy principal value, by swapping the order of summation over $N$ and $n$
in (\ref{AAA}) and replacing $N$ with $-N$ in the second term (which is equivalent to summing symmetrically about the real axis), one can obtain
\begin{align}\label{A2}
	&\mathrm{p.v.}~\textbf{A}(x,y=0)\cdot \mathbf{e}_x\nonumber\\
	=&\lim\limits_{M\rightarrow+\infty}\sum_{n=1}^{p}2\sigma_n^+\!\sum_{N=-M}^{M}\frac{y_n^++NT}{(x\!-\!x_n)^2+(y_n^+\!+\!NT)^2}.
\end{align}
Before discussing the limit expression \eqref{A2}, we first present the following theorem:
\begin{theorem} \label{the2}
	Let $z\in\mathbb{C}\backslash\mathbb{Z}$. Then
	$$\cot(\pi z)=\frac{1}{\pi z}+\frac{1}{\pi}\sum_{n=1}^{\infty}\frac{2z}{z^2-n^2}.$$
\end{theorem}
\begin{proof}
Performing the Fourier expansion of the function $\cos(zt)$ ($t\in[-\pi,\pi]$) with a period of $2\pi$, we have
$$\cos(zt)=\frac{a_0}{2}+\sum_{n=1}^{\infty}a_n\cos(nt),$$
where $a_n=\frac{1}{\pi}\int_{-\pi}^\pi\cos(zt)\cos(nt)\dd t=(-1)^n\frac{2z\sin(\pi z)}{\pi(z^2-n^2)}.$
Taking $t=\pi$ gives
$$\cot(\pi z)=\frac{1}{\pi z}+\frac{1}{\pi}\sum_{n=1}^{\infty}\frac{2z}{z^2-n^2}.$$
Hence the theorem.
\end{proof}

Define $P_M\equiv\sum\limits_{N=-M}^{M}\frac{y_n^++NT}{(x-x_n)^2+(y_n^++NT)^2}$ and $P_\infty=\lim\limits_{M\rightarrow+\infty}P_M$, suppressing the dependence on $n$ in both notations. Then \eqref{A2} can be reduced to $\mathrm{p.v.}~\textbf{A}(x,y=0)\cdot\mathbf{e}_x=2\sum_{n=1}^{p}\sigma_n^+P_\infty$. Since $\frac{b}{a^2+b^2}=\re\left(\frac{1}{b-\ii a}\right)$, one can obtain
\begin{align}\label{PM}
	P_M=\frac{1}{T}\re\left(\sum_{N=-M}^{M}\frac{1}{z^++N}\right),
\end{align}
where $z^+=\frac{y_n^+-\ii(x-x_n)}{T}.$ On the other hand, with $\frac{2z}{z^2-n^2}=\frac{1}{z-n}+\frac{1}{z+n}$, theorem \ref{the2} shows
\begin{align}\label{A3}
	\pi\cot(\pi z)=\mathrm{p.v.} \sum_{n=-\infty}^{\infty}\frac{1}{z+n}.
\end{align}
Under the Cauchy principal value, \eqref{PM} and \eqref{A3} show
\begin{align*}
	P_\infty=&\frac{\pi}{T}\re\left[\cot\left(\pi \frac{y_n^+-\ii(x-x_n)}{T}\right)\right]\\
	=&\frac{\pi}{T}\frac{\sin\frac{2\pi y_n^+}{T}}{\cosh\frac{2\pi(x-x_n)}{T}-\cos\frac{2\pi y_n^+}{T}}.
\end{align*}
Under the transformation $\xi=\frac{2\pi(x-x_n)}{T},\e^{\xi}-\cos\frac{2\pi y_n^+}{T}=\sin\frac{2\pi y_n^+}{T} w$, one obtains
\begin{align}\label{A4}
	\int_{-\infty}^\infty P_\infty\dd x=2\int_{\tan\frac{\pi y_n^+}{T}}^\infty\frac{1}{w^2+1} \dd w=\pi-\frac{2\pi y_n^+}{T}.
\end{align}
Using \eqref{A4} to integrate \eqref{A2} yields
\begin{align}\label{A5}
	&\int_{-\infty}^{\infty}\mathrm{p.v.}~\textbf{A}(x,y=0)\cdot \mathbf{e}_x\dd x\nonumber\\
	=&2\sum_{n=1}^{p}\sigma_n^+\left(\pi-\frac{2\pi y_n^+}{T}\right)=\frac{4\pi}{T}\sum_{n=1}^{p}\sigma_n^+\left(\frac{T}{2}-y_n^+\right)\nonumber\\
	=&\frac{2\pi}{T}\sum_{n=1}^{p}\left[\sigma_n^+\left(\frac{T}{2}-y_n^+\right)+\sigma_n^-\left(-\frac{T}{2}-y_n^-\right)\right],
\end{align}
which shares the same structure as \eqref{phi2}. Due to the assumptions we made regarding the summation order of the topological vector potential, this result does not fully cover the general case. Nevertheless, for the scenario involving DHBSs, numerical computations confirm that the formula holds for all examples shown in Fig. \ref{fig5}. Moreover, unlike the finite summation for RWs, the infinite summation in the soliton topological vector potential prevents swapping the order of integration and summation. This also explains why the phase shift of solitons differs from that of RWs. For RWs, we can only obtain phase shifts that are integer multiples of $\pi$, whereas for solitons, phase shifts can take values that are not integer multiples of $\pi$.

\section{Solutions of the equation $\e^{2\omega_1z}\varpi^2+\e^{-2\omega_1z}|\gamma_1|^4=\mathbf{\Gamma}\mathbf{R}=0$ in Sec.~\ref{sec5}}\label{app1}

We just consider $|\beta_2|^2=|\gamma_1|^2=\varpi$ for convenience. Under the transformation $y=2\omega_1 x$, the gradient flow $\partial_x\phi_1$ can be reduced to
\begin{align*}
	\frac{\partial \phi_1}{\partial y}
	&=\frac{P_1}{(\e^{y}+\e^{-y})P_0},
\end{align*}
where
\begin{align*}
	P_0=&(1-\omega)^2(\e^{(1+\omega)y}+\e^{-(1+\omega)y})\\
	&+(1+\omega)^2(\e^{(1-\omega)y}+\e^{-(1-\omega)y})+8\omega,\\
	P_1=&2(-2-\omega+\omega^2)(\e^{(1-\omega)y}+\e^{-(1-\omega)y})\\
	&+2(-2+\omega+\omega^2)(\e^{(1+\omega)y}+\e^{-(1+\omega)y})\\
	&+3\omega(\omega^2-1)(\e^{-2y}+\e^{2y})+2\omega(-7+3\omega^2).
\end{align*}
As $\e^{y}+\e^{-y}=0$, we have $\e^{(1-\omega)y}=-\e^{-(1+\omega)y}$ and $\e^{-(1-\omega)y}=-\e^{(1+\omega)y}.$
Therefore $P_0=0$ shows $\e^{(1-\omega)y}+\e^{-(1-\omega)y}=-2$ and $\e^{(1+\omega)y}+\e^{-(1+\omega)y}=2.$
Further we obtain $P_1=3\omega(\omega^2-1)(\e^{-2y}+\e^{2y}+2)$.
$\e^{-2y}+\e^{2y}+2=(\e^{y}+\e^{-y})^2=0$ shows $P_1=0,$ indicating that $P_0=0$ implies $P_1=0$. The solution set of the equation $\e^{y}+\e^{-y}=0$ admits two equivalent parametrizations:
\begin{align}\label{B2}
	y&=\frac{\pm\pi\ii}{2}+2\ii\pi n,\nonumber\\
 \mathrm{or}\quad y&=\frac{\pm3\pi\ii}{2}+2\ii\pi (m-1),
\end{align}
where $n,m\in\mathbb{Z}$ are independent integer indices. The two forms describe the same set of purely imaginary roots, rather than a termwise equality with matching signs and indices. Thus we focus on the pure imaginary number solutions of the equations
$$\e^{(1-\omega)y}+\e^{-(1-\omega)y}=-2,~\e^{(1+\omega)y}+\e^{-(1+\omega)y}=2.$$
The first equation gives
\begin{align}\label{B3}
	y=\frac{\pm\ii \pi+2\ii a \pi}{1-\omega},
\end{align}
where $a\in\mathbb{Z}$, while the second one gives
\begin{align}\label{B4}
	y=\frac{2\ii b \pi}{1+\omega},
\end{align}
$b\in\mathbb{Z}$. We need to find suitable $a, b$ such that the intersection of (\ref{B3}) and (\ref{B4}) is non-empty and can be represented as a subset of (\ref{B2}). Consider $\omega=p/q,~p,q\in\mathbb{Z}^+,~q>p$.
A simple construction is to make $ a=m_1 (q-p)\pm\frac{q-p-2}{4},~b=\left(m_1\pm\frac{1}{4}\right)(p+q),$
where $a,b,m_1\in\mathbb{Z}.$ This requires
\begin{align}\label{B5}
	4|q+p,~4|q-p-2.
\end{align}
In this case, (\ref{B3}) and (\ref{B4}) satisfy
$\frac{\pm\ii \pi+2\ii a \pi}{1-\omega}=\pm q\frac{\pi\ii}{2}+2qm_1\pi\ii$ and $\frac{2\ii b \pi}{1+\omega}=\pm q\frac{\pi\ii}{2}+2qm_1\pi\ii,$
respectively. Hence $q=4n_2+1$ or $q=4n_2+3,$ where $n_2\in\mathbb{Z}^+\cup\{0\}.$ Combining with (\ref{B5}), one obtains the corresponding $p$ as $4n_1+3$ and $4n_1+1$, where $n_1\in\mathbb{Z}^+\cup\{0\}.$ Thus
\begin{align}\label{B6}
	\omega=\frac{4n_1+3}{4n_2+1}~~{\rm or}~~\omega=\frac{4n_1+1}{4n_2+3}.
\end{align}
Although this selection may result in some duplicate values, it does not affect our results.

\section{Coefficients of the asymptotic DHBS solutions Eq.~\eqref{AS} before and after the collision}\label{app:collision-solution}

We provide the explicit coefficients appearing in the asymptotic expressions for the two DHBSs before and after the collision. As discussed in Sec.~V-C, the asymptotic forms of the four branches \(S_1^-\), \(S_2^-\), \(S_1^+\), and \(S_2^+\) share the common structure given by Eq.~\eqref{AS}. The coefficient sets for each branch are derived from the exact solution Eq.~\eqref{eq4} in the limits \(t\to-\infty\) and \(t\to+\infty\), and are listed below for reference.

\textbf{In the limit $t\to-\infty$: The branch $S_1^-$ ($k=1$).}
The denominator coefficients are
\begin{align*}
d_{11}^{-}={}&|\beta_4\gamma_3|^2\lambda_3^*\lambda_4^*
 |\Delta_{3\bar4}|^2|\Delta_{31}|^2|\Delta_{42}|^2,\\
d_{21}^{-}={}&\lambda_3\lambda_4|\Delta_{34}|^2
 |\Delta_{3\bar1}|^2|\Delta_{4\bar2}|^2,\nonumber\\
d_{31}^{-}={}&|\gamma_3|^2\lambda_3^*\lambda_4
 |\Delta_{3\bar4}|^2|\Delta_{31}|^2|\Delta_{4\bar2}|^2,\\
d_{41}^{-}={}&|\beta_4|^2\lambda_3\lambda_4^*
 |\Delta_{3\bar4}|^2|\Delta_{3\bar1}|^2|\Delta_{42}|^2.
\end{align*}
The numerator coefficients and prefactors are
\small
\begin{align*}
\zeta_{11}^{-}={}&\lambda_3^*\Delta_{34}|\Delta_{3\bar1}|^2,
\eta_{11}^{-}={}|\gamma_3|^2\lambda_3
 \Delta_{3\bar4}^*|\Delta_{31}|^2,\nonumber\\
\zeta_{21}^{-}={}&\lambda_4^*\Delta_{34}|\Delta_{4\bar2}|^2,
\eta_{21}^{-}={}|\beta_4|^2\lambda_4
 \Delta_{3\bar4}|\Delta_{42}|^2,\nonumber\\
c_{11}^{-}={}&
 \frac{\beta_4^*\lambda_3^*\Delta_{3\bar4}
 (\lambda_4^2-\lambda_4^{*2})}
 {\lambda_3\lambda_4\lambda_4^*}
 \frac{\lambda_2^2}{\lambda_2^{*2}}
 (\lambda_4^2-\lambda_2^{*2})(\lambda_4^{*2}-\lambda_2^{*2}),\nonumber\\
c_{21}^{-}={}&
 \frac{\gamma_3^*\lambda_4^*(\lambda_3^2-\lambda_3^{*2})
 (\lambda_3^{*2}-\lambda_4^2)}
 {\lambda_3\lambda_3^*\lambda_4}
 \frac{\lambda_1^2}{\lambda_1^{*2}}
 (\lambda_3^2-\lambda_1^{*2})(\lambda_3^{*2}-\lambda_1^{*2}).
\end{align*}
\normalsize

\textbf{In the limit $t\to-\infty$: The branch $S_2^-$ ($k=2$).}
The denominator coefficients are
\begin{align*}
d_{12}^{-}={}&|\beta_2\gamma_1|^2\lambda_1^*\lambda_2^*|\Delta_{1\bar2}|^2
 |\Delta_{1\bar3}|^2|\Delta_{1\bar4}|^2|\Delta_{2\bar3}|^2|\Delta_{2\bar4}|^2,\\
d_{22}^{-}={}&\lambda_1\lambda_2|\Delta_{12}|^2
 |\Delta_{13}|^2|\Delta_{14}|^2|\Delta_{23}|^2|\Delta_{24}|^2,\nonumber\\
d_{32}^{-}={}&|\gamma_1|^2\lambda_1^*\lambda_2|\Delta_{1\bar2}|^2
 |\Delta_{1\bar3}|^2|\Delta_{1\bar4}|^2|\Delta_{23}|^2|\Delta_{24}|^2,\\
d_{42}^{-}={}&|\beta_2|^2\lambda_1\lambda_2^*|\Delta_{1\bar2}|^2
 |\Delta_{13}|^2|\Delta_{14}|^2|\Delta_{2\bar3}|^2|\Delta_{2\bar4}|^2.
\end{align*}
The numerator coefficients and prefactors are
{\small
\begin{align*}
\zeta_{12}^{-}={}&\lambda_1^*\Delta_{12}
 |\Delta_{13}|^2|\Delta_{14}|^2,
\eta_{12}^{-}={}|\gamma_1|^2\lambda_1\Delta_{1\bar2}^*
 |\Delta_{1\bar3}|^2|\Delta_{1\bar4}|^2,\nonumber\\
 \zeta_{22}^{-}={}&\lambda_2^*\Delta_{12}
 |\Delta_{23}|^2|\Delta_{24}|^2,
\eta_{22}^{-}={}|\beta_2|^2\lambda_2\Delta_{1\bar2}
 |\Delta_{2\bar3}|^2|\Delta_{2\bar4}|^2,\nonumber\\
 c_{12}^{-}={}&
 \frac{\beta_2^*\lambda_1^*\Delta_{1\bar2}
 (\lambda_2^2-\lambda_2^{*2})}
 {\lambda_1\lambda_2\lambda_2^*\lambda_3^2\lambda_4^2}
 \lambda_3^{*2}\lambda_4^{*2}\\
 &\times
 (\lambda_2^2-\lambda_3^2)(\lambda_2^{*2}-\lambda_3^2)
 (\lambda_2^2-\lambda_4^2)(\lambda_2^{*2}-\lambda_4^2),\nonumber\\
  c_{22}^{-}={}&
 \frac{\gamma_1^*\lambda_2^*(\lambda_1^2-\lambda_1^{*2})
 (\lambda_1^{*2}-\lambda_2^2)}
 {\lambda_1\lambda_1^*\lambda_2\lambda_3^2\lambda_4^2}
 \lambda_3^{*2}\lambda_4^{*2}\\
 &\times
 (\lambda_1^2-\lambda_3^2)(\lambda_1^{*2}-\lambda_3^2)
 (\lambda_1^2-\lambda_4^2)(\lambda_1^{*2}-\lambda_4^2).
\end{align*}
}

\textbf{In the limit $t\rightarrow+\infty$: The branch $S_1^+$ ($k=1$).}
The denominator coefficients are
\begin{align*}
d_{11}^{+}={}&|\beta_4\gamma_3|^2\lambda_3^*\lambda_4^*|\Delta_{3\bar4}|^2
 |\Delta_{3\bar1}|^2|\Delta_{3\bar2}|^2|\Delta_{4\bar1}|^2|\Delta_{4\bar2}|^2,\\
d_{21}^{+}={}&\lambda_3\lambda_4|\Delta_{34}|^2
 |\Delta_{31}|^2|\Delta_{32}|^2|\Delta_{41}|^2|\Delta_{42}|^2,\nonumber\\
d_{31}^{+}={}&|\gamma_3|^2\lambda_3^*\lambda_4|\Delta_{3\bar4}|^2
 |\Delta_{3\bar1}|^2|\Delta_{3\bar2}|^2|\Delta_{41}|^2|\Delta_{42}|^2,
\end{align*}
\begin{align*}
d_{41}^{+}={}&|\beta_4|^2\lambda_3\lambda_4^*|\Delta_{3\bar4}|^2
 |\Delta_{31}|^2|\Delta_{32}|^2|\Delta_{4\bar1}|^2|\Delta_{4\bar2}|^2.
\end{align*}
The numerator coefficients and prefactors are
{\small
\begin{align*}
\zeta_{11}^{+}={}&\lambda_3^*\Delta_{34}
 |\Delta_{31}|^2|\Delta_{32}|^2,
\eta_{11}^{+}={}|\gamma_3|^2\lambda_3\Delta_{3\bar4}^*
 |\Delta_{3\bar1}|^2|\Delta_{3\bar2}|^2,\nonumber\\
\zeta_{21}^{+}={}&\lambda_4^*\Delta_{34}
 |\Delta_{41}|^2|\Delta_{42}|^2,
\eta_{21}^{+}={}|\beta_4|^2\lambda_4\Delta_{3\bar4}
 |\Delta_{4\bar1}|^2|\Delta_{4\bar2}|^2,\nonumber\\
c_{11}^{+}={}&
 \frac{\beta_4^*\lambda_3^*\Delta_{3\bar4}
 (\lambda_4^2-\lambda_4^{*2})}
 {\lambda_3\lambda_4\lambda_4^*\lambda_1^2\lambda_2^2}
 \lambda_1^{*2}\lambda_2^{*2}\\
 &\times
 (\lambda_4^2-\lambda_1^2)(\lambda_4^{*2}-\lambda_1^2)
 (\lambda_4^2-\lambda_2^2)(\lambda_4^{*2}-\lambda_2^2),\nonumber\\
c_{21}^{+}={}&
 \frac{\gamma_3^*\lambda_4^*(\lambda_3^2-\lambda_3^{*2})
 (\lambda_3^{*2}-\lambda_4^2)}
 {\lambda_3\lambda_3^*\lambda_4\lambda_1^2\lambda_2^2}
 \lambda_1^{*2}\lambda_2^{*2}\\
 &\times
 (\lambda_3^2-\lambda_1^2)(\lambda_3^{*2}-\lambda_1^2)
 (\lambda_3^2-\lambda_2^2)(\lambda_3^{*2}-\lambda_2^2).
\end{align*}
}

\textbf{In the limit $t\rightarrow+\infty$: The branch $S_2^+$ ($k=2$).}

The denominator coefficients are
\begin{align*}
d_{12}^{+}={}&|\beta_2\gamma_1|^2\lambda_1^*\lambda_2^*
 |\Delta_{1\bar2}|^2|\Delta_{13}|^2|\Delta_{24}|^2,\\
d_{22}^{+}={}&\lambda_1\lambda_2|\Delta_{12}|^2
 |\Delta_{1\bar3}|^2|\Delta_{2\bar4}|^2,\\
d_{32}^{+}={}&|\gamma_1|^2\lambda_1^*\lambda_2
 |\Delta_{1\bar2}|^2|\Delta_{13}|^2|\Delta_{2\bar4}|^2,\\
d_{42}^{+}={}&|\beta_2|^2\lambda_1\lambda_2^*
 |\Delta_{1\bar2}|^2|\Delta_{1\bar3}|^2|\Delta_{24}|^2.
\end{align*}
The numerator coefficients and prefactors are
{\small
\begin{align*}
\zeta_{12}^{+}={}&\lambda_1^*\Delta_{12}|\Delta_{1\bar3}|^2,
\eta_{12}^{+}={}|\gamma_1|^2\lambda_1
 \Delta_{1\bar2}^*|\Delta_{13}|^2,\nonumber\\
\zeta_{22}^{+}={}&\lambda_2^*\Delta_{12}|\Delta_{2\bar4}|^2,
\eta_{22}^{+}={}|\beta_2|^2\lambda_2
 \Delta_{1\bar2}|\Delta_{24}|^2,\nonumber\\
c_{12}^{+}={}&
 \frac{\beta_2^*\lambda_1^*\Delta_{1\bar2}
 (\lambda_2^2-\lambda_2^{*2})}
 {\lambda_1\lambda_2\lambda_2^*}
 \frac{\lambda_4^2}{\lambda_4^{*2}}
 (\lambda_2^2-\lambda_4^{*2})(\lambda_2^{*2}-\lambda_4^{*2}),\nonumber\\
c_{22}^{+}={}&
 \frac{\gamma_1^*\lambda_2^*(\lambda_1^2-\lambda_1^{*2})
 (\lambda_1^{*2}-\lambda_2^2)}
 {\lambda_1\lambda_1^*\lambda_2}
 \frac{\lambda_3^2}{\lambda_3^{*2}}
 (\lambda_1^2-\lambda_3^{*2})(\lambda_1^{*2}-\lambda_3^{*2}).
\end{align*}
}

\section{Virtual monopole coordinates for the asymptotic DHBSs}\label{coordinates}

We list the coordinates of the virtual monopoles for each asymptotic DHBS shown in Fig.~\ref{fig:6}.
These coordinates are numerically determined within the complex-plane regions displayed in the figure and are organized sequentially according to the asymptotic branches and the corresponding search windows. The entries below are arranged for the first component (\hyperref[fig:6]{6(a1)}-\hyperref[fig:6]{6(a4)}) and the second component (panels \hyperref[fig:6]{6(b1)}-\hyperref[fig:6]{6(b4)}), with each panel and branch identified by its heading. Within each charge-magnitude category, the upper (lower) sign in $x\pm\ii y$ or $x\mp\ii y$ corresponds to charge $+|\sigma|$ ($-|\sigma|$), with $y>0$. Thus, $x\pm\ii y$ places the positive-charge monopole in the upper half-plane, whereas $x\mp\ii y$ places it in the lower half-plane; the conjugate monopole carries the opposite charge.

\vspace{0.6em}
\textbf{Panel (a1): $S_1^+$ at $t = 1$.}
\begin{enumerate}
    \item For $|\sigma| = 1/2$: six pairs are found, which are
    \begin{align*}
    &1.913\pm\ii\,1.323,\quad 1.913\pm\ii\,0.356,\quad 1.913\pm\ii\,0.840,\\
    &1.913\mp\ii\,0.127,\quad 1.913\mp\ii\,1.094,\quad 1.913\mp\ii\,0.610.
    \end{align*}
    \item For $|\sigma| = 3/2$: ten pairs are found, which are
    \begin{align*}
    &1.772\pm\ii\,1.336,\quad 1.802\pm\ii\,0.395,\quad 1.813\pm\ii\,0.830,\\
    &2.280\pm\ii\,0.581,\quad 2.288\pm\ii\,1.339,\quad 1.774\mp\ii\,0.126,\\
    &1.788\mp\ii\,1.057,\quad 1.822\mp\ii\,0.608,\quad 2.282\mp\ii\,0.935,\\
    &2.287\mp\ii\,0.171.
    \end{align*}
\end{enumerate}

\textbf{Panel (a2): $S_2^+$ at $t=1$.}
\begin{enumerate}
    \item For $|\sigma| = 1/2$: seven pairs are found, which are
    \begin{align*}
    &-0.241\pm\ii\,0.337,\quad -0.241\pm\ii\,0.785,\quad -0.241\pm\ii\,1.234,\\
    &-0.241\mp\ii\,0.112,\quad -0.241\mp\ii\,0.561,\quad -0.241\mp\ii\,1.010,\\
    &-0.241\mp\ii\,1.459.
    \end{align*}
    \item For $|\sigma| = 3/2$: twelve pairs are found, which are
    \begin{align*}
    &-0.674\pm\ii\,0.470,\quad -0.674\pm\ii\,1.101,\quad -0.122\pm\ii\,0.785,\\
    &-0.111\pm\ii\,0.345,\quad -0.111\pm\ii\,1.225,\quad -0.675\mp\ii\,0.156,\\
    &-0.675\mp\ii\,1.414,\quad -0.673\mp\ii\,0.785,\quad -0.119\mp\ii\,0.568,\\
    &-0.119\mp\ii\,1.002,\quad -0.105\mp\ii\,0.116,\quad -0.105\mp\ii\,1.455.
    \end{align*}
\end{enumerate}

\textbf{Panel (a3): $S_1^-$ at $t=-1$.}
\begin{enumerate}
    \item For $|\sigma| = 1/2$: six pairs are found, which are
    \begin{align*}
    &-1.376\pm\ii\,0.356,\quad -1.376\pm\ii\,0.840,\quad -1.376\pm\ii\,1.323,\\
    &-1.376\mp\ii\,0.610,\quad -1.376\mp\ii\,1.094,\quad -1.376\mp\ii\,0.127.
    \end{align*}
    \item For $|\sigma| = 3/2$: ten pairs are found, which are
    \begin{align*}
    &-1.594\pm\ii\,1.338,\quad -1.504\pm\ii\,0.528,\quad -1.376\pm\ii\,0.753,\\
    &-1.248\pm\ii\,0.528,\quad -1.158\pm\ii\,1.338,\quad -1.589\mp\ii\,0.153,\\
    &-1.560\mp\ii\,0.980,\quad -1.376\mp\ii\,0.630,\quad -1.193\mp\ii\,0.980,\\
    &-1.163\mp\ii\,0.153.
    \end{align*}
\end{enumerate}

\textbf{Panel (a4): $S_2^-$ at $t=-1$.}
\begin{enumerate}
    \item For $|\sigma| = 1/2$: six pairs are found, which are
    \begin{align*}
    &0.789\pm\ii\,0.337,\quad 0.789\pm\ii\,1.234,\quad 0.789\mp\ii\,0.112,\\
    &0.789\mp\ii\,0.561,\quad 0.789\mp\ii\,1.010,\quad 0.789\mp\ii\,1.459.
    \end{align*}
    \item For $|\sigma| = 3/2$: eight pairs are found, which are
    \begin{align*}
    &0.605\pm\ii\,0.425,\quad 0.605\pm\ii\,1.146,\quad 0.972\pm\ii\,0.425,\\
    &0.972\pm\ii\,1.146,\quad 0.578\mp\ii\,0.141,\quad 0.578\mp\ii\,1.430,\\
    &1.000\mp\ii\,0.141,\quad 1.000\mp\ii\,1.430.
    \end{align*}
    \item For $|\sigma| = 5/2$: one pair is found, namely
    \begin{align*}
    &0.789\mp\ii\,0.785.
    \end{align*}
\end{enumerate}

\textbf{Panel (b1): $S_1^+$ at $t=1$.}
\begin{enumerate}
    \item For $|\sigma| = 1/2$: four pairs are found, which are
    \begin{align*}
    &2.102\pm\ii\,0.174,\quad 2.102\pm\ii\,0.931,\quad 2.102\mp\ii\,0.583,\\
    &2.102\mp\ii\,1.340.
    \end{align*}
    \item For $|\sigma| = 3/2$: ten pairs are found, which are
    \begin{align*}
    &1.772\pm\ii\,1.336,\quad 1.802\pm\ii\,0.395,\quad 1.813\pm\ii\,0.830,\\
    &2.280\pm\ii\,0.581,\quad 2.288\pm\ii\,1.339,\quad 1.774\mp\ii\,0.126,\\
    &1.788\mp\ii\,1.057,\quad 1.822\mp\ii\,0.608,\quad 2.282\mp\ii\,0.935,\\
    &2.287\mp\ii\,0.171.
    \end{align*}
\end{enumerate}

\textbf{Panel (b2): $S_2^+$ at $t=1$.}
\begin{enumerate}
    \item For $|\sigma| = 1/2$: five pairs are found, which are
    \begin{align*}
    &-0.495\pm\ii\,0.157,\quad -0.495\pm\ii\,0.785,\quad -0.495\pm\ii\,1.414,\\
    &-0.495\mp\ii\,0.471,\quad -0.495\mp\ii\,1.100.
    \end{align*}
    \item For $|\sigma| = 3/2$: twelve pairs are found, which are
    \begin{align*}
    &-0.674\pm\ii\,0.470,\quad -0.674\pm\ii\,1.101,\quad -0.122\pm\ii\,0.785,\\
    &-0.111\pm\ii\,0.345,\quad -0.111\pm\ii\,1.225,\quad -0.675\mp\ii\,0.156,\\
    &-0.675\mp\ii\,1.414,\quad -0.673\mp\ii\,0.785,\quad -0.119\mp\ii\,0.568,\\
    &-0.119\mp\ii\,1.002,\quad -0.105\mp\ii\,0.116,\quad -0.105\mp\ii\,1.455.
    \end{align*}
\end{enumerate}

\textbf{Panel (b3): $S_1^-$ at $t=-1$.}
\begin{enumerate}
    \item For $|\sigma| = 1/2$: four pairs are found, which are
    \begin{align*}
    &-1.376\pm\ii\,0.174,\quad -1.376\pm\ii\,0.931,\quad -1.376\mp\ii\,0.583,\\
    &-1.376\mp\ii\,1.340.
    \end{align*}
    \item For $|\sigma| = 3/2$: ten pairs are found, which are
    \begin{align*}
    &-1.594\pm\ii\,1.338,\quad -1.504\pm\ii\,0.528,\quad -1.376\pm\ii\,0.753,\\
    &-1.248\pm\ii\,0.528,\quad -1.158\pm\ii\,1.338,\quad -1.589\mp\ii\,0.153,\\
    &-1.560\mp\ii\,0.980,\quad -1.376\mp\ii\,0.630,\quad -1.193\mp\ii\,0.980,\\
    &-1.163\mp\ii\,0.153.
    \end{align*}
\end{enumerate}

\textbf{Panel (b4): $S_2^-$ at $t=-1$.}
\begin{enumerate}
    \item For $|\sigma| = 1/2$: four pairs are found, which are
    \begin{align*}
    &0.789\pm\ii\,0.157,\quad 0.789\pm\ii\,1.414,\quad 0.789\mp\ii\,0.471,\\
    &0.789\mp\ii\,1.100.
    \end{align*}
    \item For $|\sigma| = 3/2$: eight pairs are found, which are
    \begin{align*}
    &0.605\pm\ii\,0.425,\quad 0.605\pm\ii\,1.146,\quad 0.972\pm\ii\,0.425,\\
    &0.972\pm\ii\,1.146,\quad 0.578\mp\ii\,0.141,\quad 0.578\mp\ii\,1.430,\\
    &1.000\mp\ii\,0.141,\quad 1.000\mp\ii\,1.430.
    \end{align*}
    \item For $|\sigma| = 5/2$: one pair is found, namely
    \begin{align*}
    &0.789\mp\ii\,0.785.
    \end{align*}
\end{enumerate}


\begin{thebibliography}{99}

\bibitem{Dirac1}P. A. M. Dirac, Quantised singularities in the electromagnetic field. \newblock\href{https://doi.org/10.1098/rspa.1931.0130}{Proc. R. Soc. Lond. A \textbf{133}, 60 (1931)}.

\bibitem{Quantum1} R. P. Feynman, Space-Time Approach to Non-Relativistic Quantum Mechanics, \newblock\href{https://doi.org/10.1103/RevModPhys.20.367}{Rev. Mod. Phys. \textbf{20}, 367 (1948).}

\bibitem{Quantum2}D. Wu, Y.-F. Jiang, X.-M. Gu, L. Huang, B. Bai, Q.-C. Sun, X. Zhang, S.-Q. Gong, Y. Mao, H.-S. Zhong, M.-C. Chen, J. Zhang, Q. Zhang, C.-Y. Lu, and J.-W. Pan, Experimental Refutation of Real-Valued Quantum Mechanics under Strict Locality Conditions, \newblock\href{https://doi.org/10.1103/PhysRevLett.129.140401}{Phys. Rev. Lett. \textbf{129}, 140401 (2022).}

\bibitem{monopole1} K. A. Milton, Theoretical and experimental status of magnetic monopoles, \newblock\href{https://doi.org/10.1088/0034-4885/69/6/R02}{Rep. Prog. Phys. \textbf{69}, 1637--1711 (2006).}

\bibitem{momentum1} Z. Fang, N. Nagaosa, K. S. Takahashi, A. Asamitsu, R. Mathieu, T. Ogasawara, H. Yamada, M. Kawasaki, Y. Tokura, and K. Terakura, The anomalous Hall effect and magnetic monopoles in momentum space. \newblock\href{https://doi.org/10.1126/science.1089408}{Science \textbf{302}, 92--95 (2003).}

\bibitem{momentum2} G. E. Volovik, The Universe in a Helium Droplet (Clarendon Press, Oxford, 2003).

\bibitem{momentum3} A. B\'{e}rard, H. Mohrbach, Monopole and Berry phase in momentum space in noncommutative quantum mechanics, \newblock\href{https://doi.org/10.1103/PhysRevD.69.127701}{Phys. Rev. D \textbf{69}, 127701 (2004).}

\bibitem{momentum4} T. Dub\v{c}ek, C. J. Kennedy, L. Lu, W. Ketterle, M. Solja\v{c}i\'{c}, and H. Buljan, Weyl Points in Three-Dimensional Optical Lattices: Synthetic Magnetic Monopoles in Momentum Space,   \newblock\href{https://doi.org/10.1103/PhysRevLett.114.225301}{Phys. Rev. Lett. \textbf{114}, 225301 (2015).}

\bibitem{parameter1} M. V. Berry, Quantal phase factors accompanying adiabatic changes, \newblock\href{https://doi.org/10.1098/rspa.1984.0023}{Proc. R. Soc. London Ser. A \textbf{392}, 45--57 (1984).}

\bibitem{parameter2} D. Xiao, M.-C. Chang, and Q. Niu, Berry phase effects on electronic properties, \newblock\href{https://doi.org/10.1103/RevModPhys.82.1959}{Rev. Mod. Phys. \textbf{82}, 1959 (2010).}

\bibitem{parameter3} P. Bruno, Nonquantized Dirac Monopoles and Strings in the Berry Phase of Anisotropic Spin Systems, \newblock\href{https://doi.org/10.1103/PhysRevLett.93.247202}{Phys. Rev. Lett. \textbf{93}, 247202 (2004)}.

\bibitem{parameter4} B. Wu, Q. Zhang, J. Liu, Anomalous monopoles of an interacting boson system,  \newblock\href{https://doi.org/10.1016/j.physleta.2010.12.030}{Phys. Lett. A \textbf{375}, 545--548 (2011)}.

\bibitem{Zhao5}L.-C. Zhao, L.-Z. Meng, Y.-H. Qin, Z.-Y. Yang and J. Liu, Virtual Dirac Monopoles underlying the Nontrivial Phases of Rogue Waves. \newblock\href{https://doi.org/10.1088/0256-307X/42/11/110002}{Chin. Phys. Lett. \textbf{42}, 110002 (2025)}.

\bibitem{Zhao3}L.-C. Zhao, Y.-H. Qin, C. Lee, and J. Liu, Classification of dark solitons via topological vector potentials. \newblock\href{https://doi.org/10.1103/PhysRevE.103.L040204}{Phys. Rev. E \textbf{103}, L040204 (2021)}.

\bibitem{Zhao2}J.-D. Li, L.-Z. Meng, and L.-C. Zhao, Phase properties of several nonlinear optical waves described by rational solutions. \newblock\href{https://doi.org/10.1103/PhysRevA.107.013511}{Phys. Rev. A \textbf{107}, 013511 (2023)}.
		
\bibitem{Qin3}Y.-H. Qin, X. Zhang, L. Ling, and L.-C. Zhao, Phase characters of optical dark solitons with third-order dispersion and delayed nonlinear response. \newblock\href{https://doi.org/10.1103/PhysRevE.106.024213}
		{Phys. Rev. E \textbf{106}, 024213 (2022)}.

    	
\bibitem{Zhao1} Y.-H. Wu, L.-C. Zhao, C. Liu, Z.-Y. Yang, and W.-L. Yang, The topological phase of bright solitons. \newblock\href{https://doi.org/10.1016/j.physleta.2022.128045}{Phys. Lett. A \textbf{434}, 128045 (2022)}.

\bibitem{Zhao4}H. Yu and L.-C. Zhao, Aharonov-Anandan phase and topological vector potentials underlying an Akhmediev breather. \newblock\href{https://doi.org/10.1103/PhysRevA.110.053504}{Phys. Rev. A \textbf{110}, 053504 (2024)}.	

\bibitem{Wang2024} B.-H. Wang, N. Mao, and L.-C. Zhao, The complete set of eigenstates in one type of N-multiple quantum wells, \newblock\href{https://doi.org/10.1088/1402-4896/ad21cb}{Phys. Scr. \textbf{99}, 035108 (2024)}.
\bibitem{Sun2023} B. Sun, L.-C. Zhao, and J. Liu, Dirac magnetic monopoles potential in the nonlinear double-soliton interference, \newblock\href{https://doi.org/10.7498/aps.72.20222416}{Acta Phys. Sin. \textbf{72}, 100501 (2023)}.

\bibitem{Ray2014} M. W. Ray, E. Ruokokoski, S. Kandel, M. M\"{o}tt\"{o}nen, and D. S. Hall, Observation of Dirac monopoles in a synthetic magnetic field, \newblock\href{https://doi.org/10.1038/nature12954}{Nature \textbf{505}, 657--660 (2014).}

\bibitem{Pietila} V. Pietil\"{a} and M. M\"{o}tt\"{o}nen, Creation of Dirac monopoles in spinor Bose-Einstein condensates, \newblock\href{https://doi.org/10.1103/PhysRevLett.103.030401}{Phys. Rev. Lett. \textbf{103}, 030401 (2009).}

\bibitem{Ray2015} M. W. Ray, E. Ruokokoski, K. Tiurev, M. M\"{o}tt\"{o}nen, and D. S. Hall, Observation of isolated monopoles in a quantum field, \newblock\href{https://doi.org/10.1126/science.1258289}{Science \textbf{348}, 544--547 (2015).}

\bibitem{JHZheng} X.-Y. Chen, L. Jiang, W.-K. Bai, T. Yang, and J.-H. Zheng, Synthetic half-integer magnetic monopole and single-vortex dynamics in spherical
Bose-Einstein condensates, \newblock\href{https://doi.org/10.1103/PhysRevA.111.033322}{Phys. Rev. A \textbf{111}, 033322 (2025).}

\bibitem{ref3}A. Rogister, Parallel Propagation of Nonlinear Low-Frequency Waves in High-$\beta$ Plasma. \newblock\href{https://doi.org/10.1063/1.1693399}{Phys. Fluids \textbf{14}, 2733 (1971)}.		

\bibitem{ref1}E. Mj{\o}lhus, On the modulational instability of hydromagnetic waves parallel to the magnetic field. \newblock\href{https://doi.org/10.1017/S0022377800020249}{J. Plasma Phys. \textbf{16}, 321--334 (1976)}.

\bibitem{Mio}K. Mio, T. Ogino, K. Minami and S. Takeda, Modified nonlinear Schr\"{o}dinger equation for Alfv\'{e}n waves propagating along the magnetic field in cold plasmas. \newblock\href{https://doi.org/10.1143/JPSJ.41.265}{J. Phys. Soc. Jpn. \textbf{41}, 265 (1976)}. 			
			
\bibitem{ref2}Y. H. Ichikawa and S. Watanabe, Solitons, envelope solitons in collisionless plasmas, \newblock\href{https://doi.org/10.1051/jphyscol:1977603}{J. Phys. Colloques \textbf{38}, C6-15--C6-26 (1977)}.

\bibitem{Plasma1} F. Verheest, B. Buti, Parallel solitary Alfv\'{e}n waves in warm multi-species beam-plasma systems. Part 1, \newblock\href{https://doi.org/10.1017/S0022377800024041}{J. Plasma Phys. \textbf{47}, 15--24 (1992).}

\bibitem{Plasma2} S. R. Spangler and B. B. Plapp, Characteristics of obliquely propagating, nonlinear Alfv\'{e}n waves, \newblock\href{https://doi.org/10.1063/1.860391}{Phys. Fluids B \textbf{4}, 3356 (1992).}

\bibitem{Plasma3} B. Deconinck, P. Meuris, F. Verheest, Oblique nonlinear Alfv\'{e}n waves in strongly magnetized beam plasmas. Part 1. Nonlinear vector evolution equation, \newblock\href{https://doi.org/10.1017/S0022377800017268}{J. Plasma Phys. \textbf{50}, 445--455 (1993).}

\bibitem{Plasma4} V. Krishan, L. Nocera, Relaxed states of Alfv\'{e}nic turbulence, \newblock\href{https://doi.org/10.1016/S0375-9601(03)01105-8}{Phys. Lett. A \textbf{315}, 389--394 (2003).}

\bibitem{Plasma5} A. C.-L. Chian, W. M. Santana, E. L. Rempel, F. A. Borotto, T. Hada, and Y. Kamide, Chaos in driven Alfv\'{e}n systems: unstable periodic orbits and chaotic saddles, \newblock\href{https://doi.org/10.5194/npg-14-17-2007}{Nonlin. Processes Geophys. \textbf{14}, 17--29 (2007).}

\bibitem{optic1} D. Mihalache, N. Truta, N.-C. Panoiu and D.-M. Baboiu, Analytic method for solving the modified nonlinear Schr\"{o}dinger equation describing soliton propagation along optical fibers, \newblock\href{https://doi.org/10.1103/PhysRevA.47.3190}{Phys. Rev. A \textbf{47}, 3190 (1993).}

\bibitem{optic2} X.-J. Chen, J. Yang, Direct perturbation theory for solitons of the derivative nonlinear Schr\"{o}dinger equation and the modified nonlinear Schr\"{o}dinger equation, \newblock\href{https://doi.org/10.1103/PhysRevE.65.066608}{Phys. Rev. E \textbf{65}, 066608 (2002).}

\bibitem{optic3} N. Tzoar and M. Jain, Self-phase modulation in long-geometry optical waveguides, \newblock\href{https://doi.org/10.1103/PhysRevA.23.1266}{Phys. Rev. A \textbf{23}, 1266 (1981).}

\bibitem{optic4} K. Ohkuma, Y. H. Ichikawa, and Y. Abe, Soliton propagation along optical fibers, \newblock\href{https://doi.org/10.1364/OL.12.000516}{Opt. Lett. \textbf{12}, 516--518 (1987).}

\bibitem{FM1} L. Kavitha, M. Saravanan, V. Senthilkumar, R. Ravichandran, and D. Gopi, Collision of electromagnetic solitons in a weak ferromagnetic medium,
\newblock\href{https://doi.org/10.1016/j.jmmm.2013.11.041}{J. Magn. Magn. Mater. \textbf{355}, 37--50 (2014).}

\bibitem{FM2} M. Saravanan, Current-driven electromagnetic soliton collision in a ferromagnetic nanowire, \newblock\href{https://doi.org/10.1103/PhysRevE.92.012923}{Phys. Rev. E \textbf{92}, 012923 (2015).}

\bibitem{FM3} L. Kavitha, M. Saravanan, B. Srividya, and D. Gopi, Breatherlike electromagnetic wave propagation in an antiferromagnetic medium with Dzyaloshinsky-Moriya interaction, \newblock\href{https://doi.org/10.1103/PhysRevE.84.066608}{Phys. Rev. E \textbf{84}, 066608 (2011).}


\bibitem{Kaup} D. J. Kaup and A. C. Newell, An exact solution for a derivative nonlinear Schr\"{o}dinger equation, \newblock\href{https://doi.org/10.1063/1.523737}{J. Math. Phys. \textbf{19}, 798--801 (1978).}

\bibitem{Ling} B. Guo, L. Ling, Q. P. Liu, High-order solutions and generalized Darboux transformations of derivative nonlinear Schr\"{o}dinger equations,
               \newblock\href{https://doi.org/10.1111/j.1467-9590.2012.00568.x}{Stud. Appl. Math. \textbf{130}, 317--344 (2013)}.

\bibitem{Xu} S. Xu, J. He, L. Wang, The Darboux transformation of the derivative nonlinear Schr\"{o}dinger equation, \newblock\href{https://doi.org/10.1088/1751-8113/44/30/305203}{J. Phys. A: Math. Theor. \textbf{44}, 305203 (2011).}

\bibitem{He1}L. Guo, L. Wang, Y. Cheng, and J. He, Higher-order rogue waves and modulation instability of the two-component derivative nonlinear Schr\"{o}dinger equation. \newblock\href{https://doi.org/10.1016/j.cnsns.2019.104915}{Commun. Nonlinear Sci. Numer. Simul. \textbf{79}, 104915 (2019).}
		
\bibitem{Zuo1}Q. Liu and D. W. Zuo, Semi-rational solutions of the coupled derivative nonlinear Schr\"{o}dinger equation. \newblock\href{https://doi.org/10.1016/j.ijleo.2023.170680}{Optik \textbf{277}, 170680 (2023)}.


\bibitem{Chan} H. N. Chan, K. W. Chow, D. J. Kedziora, R. H. J. Grimshaw, and E. Ding, Rogue wave modes for a derivative nonlinear Schr\"{o}dinger model.  \newblock\href{https://doi.org/10.1103/PhysRevE.89.032914}{Phys. Rev. E \textbf{89}, 032914 (2014)}.

\bibitem{Min1} M. Li, J.-H. Xiao, W.-J. Liu, P. Wang, B. Qin, and B. Tian, Mixed-type vector solitons of the N-coupled mixed derivative nonlinear
Schr\"{o}dinger equations from optical fibers. \newblock\href{https://doi.org/10.1103/PhysRevE.87.032914}{Phys. Rev. E \textbf{87}, 032914 (2013)}.


\bibitem{Lin} H. Lin and L. Ling, Rogue wave pattern of multi-component
derivative nonlinear Schr\"{o}dinger equations, \newblock\href{https://doi.org/10.1063/5.0192741}{Chaos \textbf{34}, 043126 (2024).}

\bibitem{Akhmediev} N. Akhmediev, A. Ankiewicz, and M. Taki, Waves that appear
from nowhere and disappear without a trace, \newblock\href{https://doi.org/10.1016/j.physleta.2008.12.036}{Phys. Lett. A \textbf{373}, 675 (2009).}	

\bibitem{Gxu} G. Xu, K. Hammani, A. Chabchoub, J. M. Dudley, B. Kibler, and C. Finot, Phase evolution of Peregrine-like breathers in optics and hydrodynamics, \newblock\href{https://doi.org/10.1103/PhysRevE.99.012207}{Phys. Rev. E \textbf{99}, 012207 (2019).}

\bibitem{defocusing1} F. Baronio, M. Conforti, A. Degasperis, S. Lombardo,  M. Onorato, and S. Wabnitz, Vector Rogue Waves and Baseband Modulation Instability in the Defocusing Regime, \newblock\href{https://doi.org/10.1103/PhysRevLett.113.034101}{Phys. Rev. Lett. \textbf{113}, 034101 (2014).}
\bibitem{lingrw1} B.-L. Guo and L.-M. Ling, Rogue Wave, Breathers and Bright-Dark-Rogue Solutions for the Coupled Schr\"{o}dinger Equations, \newblock\href{https://doi.org/10.1088/0256-307X/28/11/110202}{Chin. Phys. Lett. \textbf{28}, 110202 (2011).}
\bibitem{zhaorw1} L.-C. Zhao and  J. Liu, Localized nonlinear waves in a two-mode nonlinear fiber, \newblock\href{https://doi.org/10.1364/JOSAB.29.003119}{J. Opt. Soc. Am. B \textbf{29}, 3119 (2012).}
\bibitem{zhaorw2} L.-C. Zhao and J. Liu, Rogue-wave solutions of a three-component coupled nonlinear Schr\"{o}dinger equation, \newblock\href{https://doi.org/10.1103/PhysRevE.87.013201}{Phys. Rev. E \textbf{87}, 013201 (2013).}
\bibitem{zhaorw3} L.-C. Zhao, G.-G. Xin, and Z.-Y. Yang, Rogue-wave pattern transition induced by relative frequency, \newblock\href{https://doi.org/10.1103/PhysRevE.90.022918}{Phys. Rev. E \textbf{90}, 022918 (2014).}	

\bibitem{chen1} S. Chen, C. Pan, P. Grelu, F. Baronio, and N. Akhmediev, Fundamental Peregrine Solitons of Ultrastrong Amplitude Enhancement through Self-Steepening in Vector Nonlinear Systems, \newblock\href{https://doi.org/10.1103/PhysRevLett.124.113901}{Phys. Rev. Lett. \textbf{124}, 113901 (2020).}

\bibitem{chen2} S. Chen, Y. Ye, J. M. Soto-Crespo, P. Grelu, and F. Baronio, Peregrine Solitons Beyond the Threefold Limit and Their Two-Soliton Interactions, \newblock\href{https://doi.org/10.1103/PhysRevLett.121.104101}{Phys. Rev. Lett. \textbf{121}, 104101 (2018).}


\bibitem{CDNLS1}H. C. Morris and R. K. Dodd, The Two Component Derivative Nonlinear Schr\"{o}dinger Equation. \newblock\href{https://doi.org/10.1088/0031-8949/20/3-4/029}{Phys. Scr. \textbf{20}, 505 (1979)}.	

\bibitem{Ling2} L. Ling and Q. P. Liu, Darboux transformation for a two-component derivative nonlinear Schr\"{o}dinger equation. \newblock\href{https://doi.org/10.1088/1751-8113/43/43/434023}{J. Phys. A: Math. Theor. \textbf{43}, 434023 (2010)}.	

\bibitem{Kivshar1}Y. S. Kivshar and B. Luther-Davies, Dark optical solitons: physics and applications. \newblock\href{https://doi.org/10.1016/S0370-1573(97)00073-2}{Phys. Rep. \textbf{298}, 81 (1998)}.

\bibitem{Kivshar} Y. S. Kivshar, Dark Solitons in Nonlinear Optics, \newblock\href{https://doi.org/10.1109/3.199266}{IEEE J. Quantum Electron. \textbf{29}, 250--264 (1993)}.

\bibitem{Qin2}Y.-H. Qin, L.-C. Zhao, Z.-Q. Yang, and L. Ling, Multivalley dark solitons in multicomponent Bose-Einstein condensates with repulsive interactions. \newblock\href{https://doi.org/10.1103/PhysRevE.104.014201}{Phys. Rev. E \textbf{104}, 014201 (2021)}.

\bibitem{Qin} J.-P. Yang and Y.-H. Qin, State transition dynamics of double-hump bright soliton in the coupled derivative Schr\"{o}dinger equation, \newblock\href{https://doi.org/10.1007/s11071-025-11806-9}{Nonlinear Dyn.  \textbf{113}, 32763--32781 (2025).}

\bibitem{PS1} G. Zhang and Z. Yan, The Derivative Nonlinear Schr\"{o}dinger Equation with Zero/Nonzero Boundary Conditions: Inverse Scattering Transforms and N-Double-Pole Solutions.  \newblock\href{https://doi.org/10.1007/s00332-020-09645-6} {J. Nonlinear Sci.  \textbf{30}, 3089--3127 (2020).}

\bibitem{PS2} N. Liu, Z. Xuan, and J. Sun, Triple-pole soliton solutions of the derivative nonlinear Schr\"{o}dinger equation via inverse scattering transform, \newblock\href{https://doi.org/10.1016/j.aml.2021.107741}{Appl. Math. Lett.  \textbf{125}, 107741 (2022).}
	
\bibitem{Verresen} R. Verresen, N. G. Jones,  and F. Pollmann, Topology and Edge Modes in Quantum Critical Chains, \newblock\href{https://doi.org/10.1103/PhysRevLett.120.057001}{Phys. Rev. Lett. \textbf{120}, 057001 (2018).}
\end{thebibliography}
\end{document}